\documentclass{article}
\usepackage{authblk}
\usepackage[utf8]{inputenc}
\title{Improved Approximation Factor for Adaptive Influence Maximization via Simple Greedy Strategies}
\author[1]{Gianlorenzo D'Angelo}
\author[1]{Debashmita Poddar}
\author[1]{Cosimo Vinci}
\affil[1]{Gran Sasso Science Institute, Italy}

\affil[ ]{\textit {\{gianlorenzo.dangelo,debashmita.poddar,cosimo.vinci\}@gssi.it}}
\date{}
\newcommand{\E}{\mathbb{E}}
\newcommand{\RP}{\mathbb{R}_{\geq 0}}
\usepackage{amsmath,amsfonts,amsthm,amssymb,bbm,bm}
\usepackage{appendix}
\newtheorem{theorem}{Theorem}

\newtheorem{lemma}{Lemma}
\newtheorem*{remark}{Remark}
\begin{document}
\maketitle
\begin{abstract}
In the \emph{adaptive influence maximization problem}, we are given a social network and a budget $k$, and we iteratively select $k$ nodes, called seeds, in order to maximize the expected number of nodes that are reached by an influence cascade that they generate according to a stochastic model for influence diffusion. 
The decision on the next seed to select is based on the observed cascade of previously selected seeds.
We focus on the \emph{myopic feedback model}, in which we can only observe which neighbors of previously selected seeds have been influenced and on the \emph{independent cascade} model, where each edge is associated with an independent probability of diffusing influence. While adaptive policies are strictly stronger than non-adaptive ones, in which all the seeds are selected beforehand, the latter are much easier to design and implement and they provide good approximation factors if the adaptivity gap, the ratio between the adaptive and the non-adaptive optima, is small. Previous works showed that the adaptivity gap is at most $4$, and that simple adaptive or non-adaptive greedy algorithms guarantee an approximation of $\frac{1}{4}\left(1-\frac{1}{e}\right)\approx 0.158$ for the adaptive optimum. This is the best approximation factor known so far for the adaptive influence maximization problem with myopic feedback.

In this paper, we directly analyze the approximation factor of the non-adaptive greedy algorithm, without passing through the adaptivity gap, and show an improved bound of $\frac{1}{2}\left(1-\frac{1}{e}\right)\approx 0.316$. Therefore, the adaptivity gap is at most $\frac{2e}{e-1}\approx 3.164$. To prove these bounds, we introduce a new approach to relate the greedy non-adaptive algorithm to the adaptive optimum. The new approach does not rely on multi-linear extensions or random walks on optimal decision trees, which are commonly used techniques in the field. We believe that it is of independent interest and may be used to analyze other adaptive optimization problems. Finally, we also analyze the adaptive greedy algorithm, and show that guarantees an improved approximation factor of  $1-\frac{1}{\sqrt{e}}\approx 0.393$.
\end{abstract}

\section{Introduction}
In the Influence Maximization (IM) problem, we are given a social network, a stochastic model for diffusion of influence over the network, and a budget $k$, and we ask to find a set of $k$ nodes, called \emph{seeds}, that maximize their \emph{spread of influence}, which is the expected number of nodes reached by a cascade of influence diffusion generated by the seeds according to the given diffusion model. One of the most studied models for influence diffusion is the Independent Cascade (IC), where each edge is associated with an independent probability of transmitting influence from the source node to the tail node. In the IC model the spread of influence is a monotone submodular function of the seed set, therefore a greedy algorithm, which iteratively selects a seed with maximum marginal gain,  guarantees a $1-\frac{1}{e}$ approximation factor for the IM problem~\cite{Kempe2015a}.
Since its definition by Domingos and Richardson~\cite{Domingos2001,Richardson2002} and formalization as an optimization problem by Kempe et al.~\cite{Kempe2003,Kempe2015a}, the IM problem and its variants have been extensively investigated, motivated by applications in viral marketing~\cite{Chen10}, adoption of technological innovations~\cite{Goldberg2013}, and outbreak or failure detection~\cite{Leskovec2007}. See~\cite{DBLP:series/synthesis/2013Chen,Li2018} for surveys on the IM problem.

Recently, Golovin and Krause~\cite{Golovin2011a} initiated the study of the IM problem under the framework of adaptive optimization, where, instead of selecting all the seeds at once at the beginning of the process, we can select one seed at a time and observe, to some extent, the portion of the network reached by a new selected seed. The advantage is that the decision on the next seed to choose can be based on the observed spread of previously selected seeds, usually called \emph{feedback}.
Two main feedback models have been introduced: in the \emph{full-adoption} feedback the whole spread from each seed can be observed, while, in the \emph{myopic} feedback, one can only observe the direct neighbors of each seed.

Golovin and Krause~\cite{Golovin2011a} considered the Independent Cascade model and showed that, under full-adoption feedback, the objective function satisfies the property of \emph{adaptive submodularity} (introduced in the same paper) and therefore an adaptive greedy algorithm achieves a $1-\frac{1}{e}$ approximation factor for the adaptive IM problem. In their arXiv version, they retracted a claim (appeared in their previous conference version) in which they mistakenly showed that the adaptive submodularity holds even under the myopic feedback model, and this property would have guaranteed that the adaptive greedy is a constant factor approximation algorithm (under the myopic feedback model). Anyway, they conjectured that there exists a constant factor approximation algorithm for the myopic feedback model, which indeed has been found by Peng and Chen~\cite{Peng2019} who showed that both the adaptive and non-adaptive greedy algorithms guarantee a $\frac{1}{4}\left(1-\frac{1}{e}\right)$-approximation. In particular, they showed that the adaptivity gap, which is the supremum, over all possible inputs, of the ratio between the spread of an optimal adaptive policy and that of an optimal non-adaptive one, is upper-bounded by 4. By combining this bound with the approximation factor of both the adaptive and non-adaptive greedy algorithms for the non-adaptive problem they obtain a $\frac{1}{4}\left(1-\frac{1}{e}\right)$-approximation factor. To prove their upper-bound on the adaptivity gap, Peng and Chen use an approach that is inspired by Bradac et al.~\cite{Bradac19}, which in turn is based on the approach introduced by Gupta et al.~\cite{Gupta2016,Gupta2017} in the context of stochastic probing. To relate a non-adaptive solution to an optimal adaptive one, they consider the decision tree of an optimal adaptive solution and construct a non-adaptive policy by performing a random root-leaf walk in the tree, according to a probability distribution induced by the tree. Note that computing the non-adaptive policy that guarantees the upper-bound on the adaptivity gap would require to know an optimal decision tree. However, this approach only requires to show the existence of such a non-adaptive policy since it allows us to bound the adaptivity gap and non-adaptive approximation factor is obtained by combining the non-adaptive approximation factor with the adaptivity gap. In the same paper, Peng and Chen showed that both the adaptive greedy and the non-adaptive greedy algorithms cannot achieve a factor better than $\frac{e^2+1}{(e+1)^2}<0.607<1-\frac{1}{e}$ of the adaptive optimum, and that the adaptivity gap is at least $\frac{e}{e-1}\approx 1.582$. 

In~\cite{Chen2019}, the same authors showed some upper and lower bounds on the adaptivity gap in the case of full-adoption feedback, still under independent cascade, for some particular graph classes. In order to show these bounds, they followed an approach introduced by Asadpour and Nazerzadeh~\cite{Asadpour16} which consists in transforming an adaptive policy into a non-adaptive one by means of multilinear extensions, and constructing a Poisson process to relate the influence spread of the non-adaptive policy to that of the adaptive one. For general graphs, a non-constant upper bound for the adaptivity gap under the full-adoption feedback has been recently shown in \cite{DPV21}. 

\subsection*{Our Contribution}
In this paper, we focus on the myopic model and analyze the approximation factor of the non-adaptive greedy algorithm without passing through the adaptivity gap. We show that the algorithm achieves at least a fraction of $\frac{1}{2}\left(1-\frac{1}{e}\right)\approx 0.316$ of the adaptive optimum (Theorem \ref{thm1}). By definition, this implies that the adaptivity gap is at most $\frac{2e}{e-1}\approx 3.164$ (Remark \ref{adgapcor}). For both approximation ratio and adaptivity gap we obtain a substantial improvement with respect to the upper bounds obtained in \cite{Peng2019}, which are $\frac{1}{4}\left(1-\frac{1}{e}\right)\approx 0.158$ and $4$, respectively. 

Non-adaptive policies are strictly weaker than adaptive ones, since the latter can implement the former by simply ignoring any kind of feedback. On the other hand, adaptive policies are difficult to implement as they require to probe suitable seeds and to observe the corresponding feedback, which can be expensive and error-prone. Moreover, they may consist of exponentially-large decision trees that are hard to compute and store. In contrast, non-adaptive policies are easy to design and implement and are independent from the feedback. In particular, the non-adaptive greedy algorithm has been extensively studied and successfully applied in the field of influence maximization. For the non-adaptive setting, several efficient implementation of the greedy algorithm have been devised that allows us to use it in large real-world networks~\cite{Cohen14,Goyal11,Leskovec2007,Nguyen16,Tang15,Tang2014}. Our results show that the simple non-adaptive greedy algorithm performs well, even in the adaptive setting where we compare it with the adaptive optimum.

To show our bounds, we introduce a new approach that relate the non-adaptive greedy policy to an optimal adaptive solution. The new approach is not based on multilinear extensions and poisson processes (like, e.g.~\cite{Asadpour16,Calinescu11,CVZ14,Chen2019}) neither on random walks on the optimal decision trees (like, e.g.~\cite{Bradac19,Gupta2016,Gupta2017,Peng2019}), which are the main tools used so far to relate adaptive and non-adaptive policies, and to bound adaptivity gaps. Previous techniques derive adaptive approximation factors by combining non-adaptive approximation ratios with a bound on the adaptivity gap which is obtained by showing the existence of a ``good'' non-adaptive policy. However such a policy is hard to compute as it is usually  constructed by using an optimal adaptive policy. Our approach, instead, directly analyzes a non-adaptive policy and therefore provides the exact policy that gives the desired adaptivity gap and adaptive approximation factor. We believe that our approach is of independent interest and may be used to bound approximation factors and adaptivity gaps of different adaptive optimization problems. 

Our new approach consists in defining a simple randomized non-adaptive policy whose performance is not higher than that guaranteed by the greedy algorithm, and to relate such randomized non-adaptive policy with the optimal adaptive policy. In order to recover good properties of the objective function (like, e.g. submodularity) that usually guarantee good approximations when adopting  greedy strategies, we introduce an artificial diffusion process in which each seed has two chances to influence its neighbours. A similar process was introduced by Peng and Chen~\cite{Peng2019}, who consider a diffusion model in which the seeds appear in multiple copies of the influence graphs, so that, roughly speaking, each node has several chances to influence the neighbours, and the main machinery they consider to relate optimal adaptive strategies with optimal non-adaptive ones is that in~\cite{Bradac19}.
Our direct and more refined analysis of the non-adaptive greedy algorithm improves  at the same time both the approximation ratio and the adaptivity gap. 

To illustrate our approach, in Section~\ref{sec_example} we first apply our machinery to the simpler setting of adaptive monotone submodular maximization under cardinality constraint. As observed by Asadpour and Nazerzadeh~\cite{Asadpour16}, a constant factor approximation of the non-adaptive greedy policy applied to such setting can be obtained by combining the approximation ratio over the non-adaptive optimum and the adaptivity gap (which is equal to $\left(1-\frac{1}{e}\right)$~\cite{Asadpour16}); this leads to an approximation guarantee of $\left(1-\frac{1}{e}\right)^2\approx 0.399$. We~give a more refined analysis of the non-adaptive greedy policy and we show that its approximation ratio is at least $\frac{1}{2}\left(1-\frac{1}{e^2} \right)\approx 0.432$. Asadpour and Nazerzadeh \cite{Asadpour16} also showed that a so-called continuous greedy policy \cite{Calinescu11,CVZ14} achieves an approximation ratio of $1-\frac{1}{e}-\epsilon$ (in polynomial time w.r.t.~$\frac{1}{\epsilon}$). Since the continuous greedy policy is a non-adaptive policy, this bound is strictly better than ours. However, the non-adaptive greedy policy is simpler and deterministic, while the continuous greedy policy is randomized and more complex.

Finally, in Section \ref{ad-sec}, we analyze the adaptive version of the greedy algorithm applied to the adaptive influence maximization problem; by resorting again to an artificial diffusion process, we show that such adaptive algorithm guarantees an approximation ratio of $1-\frac{1}{\sqrt{e}} ~\approx~0.393$; thus, we further improve the upper bound shown for the non-adaptive greedy algorithm, and we also give a more refined analysis of the adaptive greedy algorithm than that of Peng and Chen \cite{Peng2019}, who showed a $\frac{1}{4}\left(1-\frac{1}{e}\right)\approx 0.158$ approximation factor.

\subsection*{Related Work}
\subparagraph*{Adaptive Influence Maximization.}
The adaptive influence maximization problem under the independent cascade model has been studied by~\cite{Chen2019, Chen2019a,  Han2018a,Peng2019,Sun2018, Tang2019, Tong2019, Tong2017,  DBLP:journals/corr/VaswaniL16, Yuan2017}. These include studies on several classes of graphs and different feedback models. 

The most studied feedback model is the full-adoption feedback, in which the entire influence spread generated by each selected node is observed. Golovin and Krause~\cite{Golovin2011a} show that the full-adoption feedback model satisfies the adaptive submodularity property and, by exploiting such property, provide a $(1-1/e)$ approximation adaptive algorithm for the problem of finding the best adaptive policy. Chen and Peng~\cite{Chen2019} study the adaptivity gap in the full-adoption feedback model under certain restrictions on the graph topology. In particular, they show that the adaptivity gap of in-arborescence (resp. out-arborescence) graphs belongs to $[e/(e-1),2e/(e-1)]$ (resp. $[e/(e-1),2]$). For general graphs with $n$ nodes (resp. in-arborescence graphs), an upper bound of $O(n^{1/3})$ (resp. $2e^2/(e^2-1)$) for the adaptivity gap under the full-adoption feedback has been recently shown in \cite{DPV21}.

Golovin and Krause~\cite{Golovin2011a} conjectured that the influence maximization problem under the myopic feedback model admits a constant approximation algorithm and a constant adaptivity gap, despite the adaptive submodularity does not hold under such feedback model. Since then, several studies have been conducted on the myopic feedback model. Some recent works include that of Salha et al.~\cite{Salha2018}, in which they consider a modified version of the independent cascade model which gives multiple chances to the seeds to activate their  neighbours, and  consider a different utility function which needs to maximized. They demonstrate that the myopic feedback model is adaptive submodular under such modified diffusion model, and provide an adaptive greedy policy that achieves a $1-1/e$ approximation ratio for the problem of finding the best adaptive policy. The work of Peng and Chen~\cite{Peng2019} is the first one that provides a constant upper bound on the adaptivity gap under the myopic feedback model. They introduce a policy in which each seed appears in multiple copies of the original graph; furthermore, this hybrid policy connects the adaptive and the non-adaptive policies via a  machinery used by \cite{Gupta2016,Gupta2017,Bradac19} in the context of stochastic probing. They show that the adaptivity gap is in $\left[e/(e-1),4\right]$, and the upper bound is turned into $(1-1/e)/4$ approximation algorithm. 

Other diffusion and feedback models have been also studied, e.g., the multi-round diffusion model \cite{Sun2018}, the general feedback model \cite{Tong2019}, and the partial feedback model \cite{Yuan2017}. Han et al.~\cite{Han2018a} conduct a study on the batch selection of seeds at each step of the diffusion process. 
Tong et al.~\cite{Tong2015} introduce the dynamic independent cascade model, which captures the dynamic nature of real-world social networks. Finally, Singer et al.~\cite{Badanidiyuru2016,Rubinstein2015, Seeman2013, Singer2016}, in their line of research on adaptivity gaps, studied a two-stage process called adaptive seeding, that exhibits some similarities with the influence maximization problem under the myopic feedback model. 

\subparagraph*{Other Topics in Adaptive Optimization.}
Beyond influence maximization problems, the adaptive optimization and the adaptivity gap have been generally studied for many other stochastic settings~\cite{DBLP:conf/stacs/AdamczykSW14,Asadpour16, DBLP:conf/wine/AsadpourNS08, Bradac19,DBLP:journals/mor/ChanF09,   DBLP:conf/soda/DeanGV05, DBLP:journals/mor/DeanGV08,DBLP:conf/latin/GoemansV06,   Golovin2011a,DBLP:conf/icml/GuilloryB10, Gupta2016,Gupta2017,DBLP:conf/ciac/HellersteinKL15,DBLP:journals/corr/abs-1803-07639,Bradac19}.

A general adaptive optimization framework deals with the fact that an item will reveal its actual state only when it has been irrevocably included in the final solution, and the main goal is to optimize an objective function under such uncertainty. Stochastic variants of packing integer programs, 0/1 knapsack problems, and covering problems, have been studied under the perspective of adaptive optimization in \cite{DBLP:conf/soda/DeanGV05, DBLP:journals/mor/DeanGV08,DBLP:conf/latin/GoemansV06}, respectively.

Asadpour et al.~\cite{Asadpour16, DBLP:conf/wine/AsadpourNS08} study the adaptivity gap of the stochastic submodular maximization problem under a matroid constraint, in which the goal is to select a subset of items satisfying a matroid constraint, that maximizes the value of a monotone submodular value function defined on the random states of the selected items. In~\cite{Asadpour16}, the authors consider an adaptive greedy policy (often denoted as myopic policy) to approximate the optimal value of the best adaptive policy, and they show that it achieves a $1/2$ approximation ratio under general matroid constraint, and a $1-\frac{1}{e}$ approximation ratio  under cardinality constraint. Interesting variants or extensions of the above optimization problem have been  considered in \cite{Bradac19,Gupta2016,Gupta2017}.

Chan and Farias~\cite{DBLP:journals/mor/ChanF09} study the efficiency of adaptive greedy policies applied to a general class of stochastic optimization problems, called stochastic depletion problems, in which the adaptive policy, at each step, chooses an action that generates a reward and depletes some of the items; they show that, under certain structural properties, a simple adaptive greedy policy guarantees a constant factor approximation of the best adaptive policy. Hellerstein et al.~\cite{DBLP:conf/ciac/HellersteinKL15} use an optimal decision tree to build a connection between the adaptive and the non-adaptive policies, and show that the adaptivity gap of stochastic submodular maximization under cardinality constraint is $1-\frac{1}{e^\tau}$, where $\tau$ is the minimum value of the probability that an  item is in some state.

\subsection*{Organization of the Paper}
In the next section we introduce our new approach by applying it to a simpler setting. In Section~\ref{sec_prel} we introduce the adaptive influence maximization problem along with the necessary notation and definitions. In Section~\ref{sec_inarb} we give the main results of the paper, that is the improved approximation factor for the non-adaptive greedy algorithm and the upper bound on the adaptivity gap for adaptive influence maximization. In Section~\ref{ad-sec} we show an improved approximation ratio for the adaptive greedy algorithm. In Section~\ref{sec_future} we outline possible research directions. Finally, an appendix contains the proofs of the technical lemmas which the main theorems are based on.

\section{Overview of the Approach: Stochastic Submodular Maximization}\label{sec_example}
In this section, we illustrate our machinery by applying part of it to the problem of maximizing a stochastic submodular set function under cardinality constraints~\cite{Asadpour16}. 

For two integers $h$ and $k$, $h\leq k$, let $[k]_h:=\{h,h+1,\ldots, k\}$ and $[k]:=[k]_1$. A function $f:\RP^n\rightarrow \RP$ is {\em a monotone submodular value function} if, for any vectors $x,y\in \RP$, we get $f(x\vee y)+f(x\wedge y)\leq f(x)+f(y)$, where $x\wedge y$ denotes the componentwise minimum and $x\vee y$ denotes the componentwise maximum.

Let $[n]$ be a finite set of $n$  items, and let $\bm \theta=(\bm \theta_1,\ldots, \bm \theta_n)$ be a vector of $n$ real, non-negative, and independent {\em state} random variables, where each $\bm \theta_i$ returns the state $\theta_i\in \RP$ associated to each item $i$ following a certain probability distribution. Let $f:\RP^n\rightarrow \RP$ be the {\em objective function}, that is a monotone submodular value function. For any $S\subseteq [n]$, let $\bm \theta(S):=(\bm \theta_1(S),\ldots, \bm \theta_n(S))$ be the {\em partial state} random variable, that is a random vector defined as $\bm \theta_i(S)=\bm \theta_i$ if $i\in S$, $\bm \theta_i(S)=0$ otherwise. With a little abuse of notation, we assume that vector $\bm \theta(S)$ gives also information on the set $S$ which $\bm \theta(S)$ is based on.

For a given integer $k\geq 0$, we aim at selecting a subset $S\subseteq [n]$ subject to cardinality constraint $|S|=k$, that maximizes the expected value $\E_{\bm \theta}[f(\bm \theta(S))]$. To guarantee a (possibly) better solution we can resort to an {\em adaptive policy}, that, at each step, observes the partial state $\xi\sim\bm \theta(U)$, where $U$ denotes the set of items previously selected, and selects another further item $\pi(\xi)\in [n]\setminus U$;  after $k$ iterations, the policy returns a set $U_{\bm\theta,k}(\pi)\subseteq [n]$ with $|U_{\bm\theta,k}(\pi)|=k$, which is a random set depending on the state of $\bm \theta$. The {\em stochastic monotone submodular maximization problem} (SMSM) takes as input a set of items $[n]$, a random vector~$\bm \theta$, a monotone submodular value function $f$, and an integer $k\in [n]$, and asks to find an adaptive policy $\pi$ that maximizes the expected value $\E_{\bm \theta}[f(\bm \theta(U_{\bm\theta,k}(\pi)))]$. 

In general, computing an optimal adaptive strategy is a computationally hard problem. Furthermore, in many contexts it is difficult to implement adaptive strategies, and non-adaptive strategies (in which the solution is chosen without observing the states of the random variables) is a more feasible choice. Asadpour and Nazerzadeh \cite{Asadpour16} show that a non-adaptive randomized continuous greedy algorithm guarantees a $\left(1-\frac{1}{e}-\epsilon\right)$ approximation for the SMSM problem (in polynomial time w.r.t. $1/\epsilon$); however, the proposed approach resorts to a quite sophisticated randomized algorithm. They also consider a simpler and deterministic {\em non-adaptive greedy algorithm} as approximation algorithm, that starts from an empty set $S:=\emptyset$ and, at each iteration $t\in [k]$, adds in $S$ the item $i\in [n]\setminus S$ maximizing $\E_{\bm \theta}[f(\bm \theta(S\cup\{i\}))]$. They show that the greedy algorithm guarantees an approximation factor of $\frac{1}{2}\left(1-\frac{1}{e}\right)\approx 0.316$ if the chosen subsets are subject to a matroid constraint, that can be reduced to $\left(1-\frac{1}{e}\right)^2\approx 0.399$ when considering a cardinality constraint $|S|\leq k$ only (i.e., uniform matroid constraint). 

In the following theorem, we give a better analysis of the greedy algorithm under cardinality constraints, and we show that the approximation factor increases to $\frac{1}{2}\left(1-\frac{1}{e^2}\right)\approx 0.432$. 

\begin{theorem}\label{thm0}
The non-adaptive greedy algorithm is a $\frac{1}{2}\left(1-\frac{1}{e^2}\right)$ approximation algorithm for the SMSM problem. 
\end{theorem}
For the proof of Theorem \ref{thm0} we relate the non-adaptive greedy solution with the optimal adaptive solution. We assume w.l.o.g that $k\geq 2$, otherwise the approximation ratio is equal to~$1$. For any $t\in [k]_0$, let $S_t$ be the set of $t$ items computed by the greedy algorithm at iteration $t$ (where $S_0:=\emptyset$). Let $ \pi$ be an optimal adaptive policy, and let $x=(x_1,\ldots, x_n)\in [0,1]^n$ be the vector such that $x_i$ is the probability that node $i\in [n]$ is selected by policy $\pi$. Let $OPT_A(k)$ denote the value of policy $\pi$ (i.e., the optimal value of the problem). Given $i\in [n]$, let $\bm e^{i}:=(\bm e^{i}_1,\ldots, \bm e^{i}_n)$ be a random vector where $\bm e^{i}_j$ has the same distribution of $\bm\theta_i$ if $i=j$, and $\bm e^{i}_j=0$ otherwise; given a partial state $\xi$, let $\Delta(i|\xi):=\E_{\bm e^i}[f(\xi \vee \bm e^{i})-f(\xi)],$ i.e., the expected increment of the objective function when adding an item $i$ under partial state $\xi$.

For any $t\in [k-1]_0$, as intermediate step of our analysis, we consider a {\em randomized non-adaptive policy} ${\sf Rand}_t$ that, starting from the greedy solution $S_{t}$, computes a random set $S_{\bm\rho,t}:=S_{t}\cup\{\bm \rho\}$, where $\bm\rho\in [n]$ is a random item such that $\mathbb{P}[\bm\rho=i]=x_i/k$ for any $i\in [n]$ and selected independently from any other event. Observe that the above random variable is well-defined, as $\sum_{v\in V}(x_v/k)=k/k=1$; furthermore, we observe that the expected value of $f$ under ${\sf Rand}_{t}$ is $\E_{\bm\theta,\bm\rho}[f(\bm \theta(S_{\bm\rho,t})]$. Furthermore, for any $t\in [k-1]_0$, we consider a {\em hybrid adaptive policy} ${\sf Hyb}_{t}$, that first runs the adaptive policy $\pi$, and then merges the items of $S_t$ with the items of $U_{\hat{\bm\theta},k}(\pi)$ selected by $\pi$, where $\hat{\bm \theta}$ is a random state that follows the same distribution of $\bm \theta$ but is  independent from $\bm \theta$; finally, the expected value of $f$ under ${\sf Hyb}_{t}$ is defined as $\E_{\bm\theta,\hat{\bm \theta}}[f(\bm\theta(S_t)\vee\hat{\bm \theta}(U_{\hat{\bm\theta},k}(\pi)\cup S_t))].$

A similar hybrid adaptive policy has been also considered by Asadpour and Nazerzadeh \cite{Asadpour16}, but in place of the randomized non-adaptive policy considered in our work, they resort to a non-adaptive strategy defined by a Poisson process based on the multilinear extension of the expected value function $\E_{\bm \theta}[f(*)]$. Before showing the theorem, we give some preliminary lemmas, which relate the randomized non-adaptive policy with the hybrid adaptive policy.

\begin{lemma}\label{lemprel1}
We have that $\E_{\bm \theta,\bm \rho}[f(\bm\theta(S\cup\{\bm \rho\}))]-\E_{\bm \theta}[f(\bm\theta(S))]= \sum_{i\in [n]\setminus S}x_i\cdot \E_{\bm\theta}[\Delta(i|\bm\theta(S))]$ for any $S\subseteq [n]$.
\end{lemma}
\begin{lemma}\label{lemprel2}
For any $S\subseteq [n]$, we have $\E_{\bm\theta,\hat{\bm \theta}}[f(\bm\theta(S)\vee\hat{\bm \theta}(U_{\hat{\bm\theta},k}(\pi)\cup S))]\\
\leq \E_{\bm\theta,\hat{\bm \theta}}[f( \bm\theta(S)~\vee~\hat{\bm \theta}(S)]+\sum_{i\in [n]\setminus S}x_i\cdot \E_{\bm\theta}[\Delta(i|\bm\theta(S))]$.
\end{lemma}
\begin{lemma}\label{lemprel3}
We have that $OPT_A(k)\leq \E_{\bm\theta,\hat{\bm \theta}}[f(\bm\theta(S)\vee \hat{\bm \theta}(U_{\hat{\bm\theta},k}(\pi)\cup S))]$ for any $S\subseteq [n]$.
\end{lemma}
Armed with the above lemmas, we can prove Theorem \ref{thm0}.
\begin{proof}[Proof of Theorem \ref{thm0}]
For any $t\in [k]_0$, let $GR_N(t):=\E_{\bm \theta}[f(\bm \theta(S_t))]$ denote the expected value of $f$ obtained at the $t$-th iteration of the non-adaptive greedy algorithm. We have that
\begin{align}
&GR_N(t+1)-GR_N(t)=\max_{v\in V}\left[\E_{\bm \theta}[f(\bm\theta(\{v\}\cup S_t))]\right]-\E_{\bm \theta}[f(\bm\theta(S_t))]\label{equ0prel}\\
&\geq \overbrace{\E_{\bm \theta,\bm \rho}[f(\bm\theta(S_{t}\cup\{\bm \rho\}))]}^{\text{Exp. value of }{\sf Rand}_t}-\E_{\bm \theta}[f(\bm\theta(S_t))]\nonumber\\
&\geq \frac{1}{k}\cdot \sum_{i\in [n]\setminus S_t}x_i\cdot \E_{\bm\theta}[\Delta(i|\bm\theta(S_t))]\label{equ1prel}\\
&\geq \frac{1}{k}\cdot\left(\overbrace{\E_{\bm\theta,\hat{\bm \theta}}[f(\bm\theta(S_t)\vee \hat{\bm \theta}(U_{\hat{\bm\theta},k}(\pi)\cup S_t))]}^{\text{Exp. value of }{\sf Hyb}_t}-\E_{\bm\theta,\hat{\bm \theta}}[f( \bm\theta(S_t)\vee\hat{\bm \theta}(S_t))]\right)\label{equ2prel}\\
&\geq \frac{1}{k}\cdot\left(OPT_A(k)-\E_{\bm\theta,\hat{\bm \theta}}[f( \bm\theta(S_t)\vee\hat{\bm \theta}(S_t))]\right)\label{equ3prel}\\
&\geq \frac{1}{k}\cdot OPT_A(k)-\frac{1}{k}\cdot \left(\E_{\bm\theta,\hat{\bm \theta}}[f( \bm\theta(S_t))+f( \hat{\bm\theta}(S_t))]\right)\label{equ4prel}\\ 
&\geq \frac{1}{k}\cdot OPT_A(k)-\frac{1}{k}\cdot \left(\E_{\bm\theta}[f( \bm\theta(S_t))]+\E_{\hat{\bm \theta}}[f( \hat{\bm\theta}(S_t))]\right)\nonumber\\
&= \frac{1}{k}\cdot OPT_A(k)
-\frac{2}{k}\cdot \E_{\bm\theta}[f(\bm\theta(S_t)]\nonumber\\
&=\frac{1}{k}\cdot OPT_A(k)-\frac{2}{k}\cdot GR_N(t),\label{equfundprel}
\end{align}
where \eqref{equ0prel} comes from the fact that the greedy strategy, at iteration $t+1$, adds to $S_t$ the item $i$ guaranteeing the best expected value of $f$, \eqref{equ1prel} comes from Lemma \ref{lemprel1}, \eqref{equ2prel} comes from Lemma \ref{lemprel2}, \eqref{equ3prel} comes from Lemma \ref{lemprel3}, and \eqref{equ4prel} holds since $f$ is a monotone submodular value function (as $f(\bm\theta(S_t)\vee\hat{\bm \theta}(S_t))\leq f(\bm\theta(S_t)\vee\hat{\bm \theta}(S_t))+f(\bm\theta(S_t)\wedge\hat{\bm \theta}(S_t))\leq f( \bm\theta(S_t))+f( \hat{\bm\theta}(S_t))$). Thus, by \eqref{equfundprel} and some manipulations, we get $
{GR}_N(t+1)\geq \frac{1}{k}\cdot OPT_A(k)+\left(1-\frac{2}{k}\right)\cdot {GR}_N(t)$ for any $t\in [k-1]_0.$ By applying iteratively the above inequality, we get $
{GR}_N(k)\geq \frac{1}{k}\cdot \sum_{t=0}^{k-1}\left(1-\frac{2}{k}\right)^{t}\cdot OPT_A(k)=\frac{1}{2}\left(1-\left(1-\frac{2}{k}\right)^{k}\right)\cdot OPT_A(k),$ that leads to $\frac{{GR}_N(k)}{OPT_A(k)}\geq \frac{1}{2}\left(1-\left(1-\frac{2}{k}\right)^{k}\right) \geq \frac{1}{2}\left(1-\frac{1}{e^2}\right),$ and this shows the claim. 
\end{proof}
\section{Adaptive Influence Maximization under the Myopic Feedback Model: Preliminaries}\label{sec_prel}
\subparagraph*{Independent Cascade Model.} In the {\em independent cascade model} (IC), we have an {\em influence graph}  $G=(V=[n],E,(p_{uv})_{(u,v)\in E})$, where edges are directed and $p_{uv}\in [0,1]$ is an {\em activation probability} associated to each edge $(u,v)\in E$. Given a set of {\em seeds} $S\subseteq V$ which are initially \emph{active}, the diffusion process in the IC model is defined in $t\geq 0$  discrete steps as follows: (i) let $A_t$ be the set of active nodes which are activated at each step $t\geq 0$; (ii) $A_0:=S$; (iii) given a step $t\geq 0$, for any edge $(u,v)$ such that $u\in A_t$, node $u$ can activate node $v$ with probability $p_{uv}$ independently from any other node, and, in case of success, $v$ is included in $A_{t+1}$; (iv) the diffusion process ends at a step $r\geq 0$ such that $A_{r}=\emptyset$, i.e., no node can be activated at all. The size of $\bigcup_{t\leq r} {A_t}$, i.e. the number of nodes activated/reached by the diffusion process, is the {\em influence spread}. 

The above diffusion process can be equivalently defined as follows. The {\em live-edge graph} $\bm L=(V,\bm L(E))$ is a random graph made from $G$, such that each edge $(u,v)\in E$ is included in $\bm L(E)$ with probability $p_{uv}$, independently from the other edges, i.e., $\mathbb{P}[\bm L=L]=\prod_{(u,v)\in L}p_{uv}\prod_{(u,v)\in E\setminus L}(1-p_{uv})$. With a little abuse of notation, we may denote $\bm L(E)$ with $\bm L$. Given $L\subseteq E$, let $R_L(S):=\{v\in V:\text{ there exists a path from $u$ to $v$ in $L$ for some $u\in S$}\}$, i.e., the set of nodes reached by nodes in $S$ in the graph $L$. Informally, if $S$ is the set of seeds, and $L$ is a realisation of the live-edge graph, $R_L(S)$ equivalently denotes the set of nodes which are reached/activated by the above diffusion process. Let $\sigma_L(S):=|R_L(S)|$ denote the {\em influence spread} generated by the set of seeds $S$ if the realised live-edge graph is $L$, and let $\sigma(S):=\mathbb{E}_{\bm L}[\sigma_{{\bm L}}(S)]$ be the {\em expected influence spread} generated by $S$. 

\subparagraph*{Non-adaptive Influence Maximization.} The {\em non-adaptive influence maximization problem under the IC model} is the computational problem that, given an influence graph $G$ and an integer $k\geq 1$, asks to find a set of seeds $S\subseteq V$ with $|S|\leq k$ such that $\sigma(S)$ is maximized. Without loss of generality, we assume that $k\in [n]$ and, since the objective function is monotone, $|S|=k$ for any solution $S$.

Kempe et al.~\cite{Kempe2003,Kempe2015a} showed that function $\sigma$ is monotone and submodular, therefore the following greedy algorithm achieves a $1-\frac{1}{e}$ approximation factor: (i) start with an empty set of seeds $S:=\emptyset$; (ii) at each iteration $t\in [k]$, add to $S$ the node $v$ that maximizes the expected influence spread $\sigma(S\cup \{v\})$. Note that the greedy algorithm requires at each iteration to compute the value of function $\sigma$, for some set of seeds and this has been shown to be computationally intractable as it is $\#P$-hard~\cite{Chen10}. However, standard Chernoff bounds allows us estimate the value of $\sigma$ through a polynomial number of Monte-Carlo simulations by introducing an arbitrarily small additive error $\epsilon>0$, which depends on the number of simulations~\cite{Kempe2015a}. In the reminder of the paper we will omit the additional term $\epsilon$ to avoid unnecessary complicated formulas. We will refer to this algorithm as the \emph{non-adaptive greedy algorithm}.

\subparagraph*{Adaptive Influence Maximization.} Differently from the non-adaptive setting, in which all the seeds are selected at the beginning, an {\em adaptive policy} activates the seeds sequentially in $k$ steps,
one seed at each step, and the decision on the next seed to select is based on the feedback resulting from the observed spread of previously selected nodes. The feedback model considered in this work is {\em myopic}: when a node is selected, the adaptive policy observes the state of its neighbours. 

An adaptive policy under the myopic feedback model is formally defined as follows. Given $L\subseteq E$, the {\em realisation} $\phi_L:V\rightarrow 2^V$ associated to $L$ assigns to each node $v\in V$ the value $\{z\in V:(v,z)\in L\}\cup\{v\}$, i.e., the set containing $v$ and the neighbours activated by seed $v$ when $\bm L=L$. Let $\Phi$ denote the {\em random realisation}, i.e., the random variable such that $\mathbb{P}[\Phi=\phi_L]=\mathbb{P}[\bm L=L]$ for any $L\subseteq E$. Given a set $S\subseteq V$, a {\em partial realisation} $\psi:S\rightarrow 2^V$ is the restriction to $S$ of the domain of some realisation, i.e., there exists $L\subseteq E$ such that $\psi(v)=\phi_L(v)$ for any $v\in S$. Given a partial realisation $\psi:S\rightarrow 2^V$, let $dom(\psi):=S$, i.e., $dom(\psi)$ is the domain of partial realisation $\psi$, and let $Im(\psi):=\bigcup_{v\in dom(\psi)}\psi(v)$. A~partial realisation $\psi'$ is a {\em sub-realisation} of a partial realisation $\psi$ (or, equivalently,  $\psi'\subseteq \psi$), if $dom(\psi')\subseteq dom(\psi)$ and $\psi'(v)=\psi(v)$ for any $v\in dom(\psi')$. We observe that any partial realisation $\psi$ can be equivalently represented as $\{(v,\phi_L(v)):v\in dom(\psi)\}$ for some $L\subseteq E$.

An adaptive policy $\pi$ takes as input a partial realisation $\psi$ and, either returns a node $\pi(\psi)\in V$ and activates it as seed, or interrupts the activation of new seeds, e.g., by returning a string $\pi(\psi):=STOP$. In particular, an adaptive policy $\pi$ can be run as follows: (i) start from an empty realisation $\psi:=\emptyset$; (ii) if $\pi(\psi)\neq STOP$ set $\psi\leftarrow \psi\cup \{(v,\phi_L(v))\}$ and repeat (ii) until $\pi(\psi)=STOP$; (iii) at the end,  return $\psi_{\pi}:=\psi$. Let $\Psi_{\pi}$ be the random partial realisation returned by the execution of policy $\pi$. The {\em expected influence spread} of an adaptive policy $\pi$ is defined as $\sigma(\pi):=\mathbb{E}_{\bm L}[\sigma_{\bm L}(dom(\Psi_{\pi}))]$, i.e., it is the expected value of the number of nodes reached by the diffusion process at the end of the execution of policy $\pi$. We say that $|\pi|=k$ if policy $\pi$ always return a partial realisation $\psi_{\pi}$ with $|dom(\psi_{\pi})|=k$. The {\em adaptive influence maximization problem (under IC and myopic feedback)} is the computational problem that, given an influence graph $G$ and $k\in [n]$, asks to find an adaptive policy $\pi$ subject to $|\pi|=k$ that maximizes the expected influence spread $\sigma(\pi)$. 

\subparagraph*{Adaptivity gap.}Given an influence graph $G$ and an integer $k\geq 1$, let $OPT_N(G,k)$ (resp. $OPT_A(G,k)$) denote the optimal value of the non-adaptive (resp. adaptive) influence maximization problem with input $G$ and $k$. Given an integer $k\in [n]$, the {\em $k$-adaptivity gap} of $G$ is defined as $AG(G,k):=\frac{OPT_A(G,k)}{OPT_N(G,k)},$ and measures how much an adaptive policy outperforms a non-adaptive solution for the influence maximization problem applied to influence graph $G$, when the maximum number of seeds is $k$. The {\em adaptivity gap} of $G$ is defined as $AG(G):=\sup_{k\in [n]}AG(G,k)$. We observe that for $k=1$ the $k$-adaptivity gap and the approximation factors are trivially equal to 1, thus we omit such case in the following.

\section{The Efficiency of the Non-adaptive Greedy Algorithm}\label{sec_inarb}
In this section, we show that a simple non-adaptive algorithm guarantees an approximation ratio of $\frac{1}{2}\left(1-\frac{1}{e}\right)\approx 0.316$ for the adaptive influence maximization problem, thus improving the approximation ratio of $\frac{1}{4}\left(1-\frac{1}{e}\right)\approx 0.158$ given in~\cite{Peng2019}. The algorithm provided by~\cite{Peng2019} is the usual {\em non-adaptive greedy algorithm} given in~\cite{Kempe2015a} and reported in the previous section.

We observe that such algorithm is non-adaptive, i.e., despite it is used for adaptive optimization, does not resort to the use of any adaptive policy and all the seeds are selected without observing any partial realisation. 
\begin{theorem}\label{thm1}
Given an influence graph $G$ with $n$ nodes and $k\in [n]_2$, the non-adaptive greedy algorithm is a $\frac{1}{2}\left(1-\left(1-\frac{1}{k}\right)^{k}\right)\geq \frac{1}{2}\left(1-\frac{1}{e}\right)$ approximation algorithm for the adaptive influence maximization problem (under IC and myopic feedback) applied to $(G,k)$. 
\end{theorem}
In the proof of Theorem~\ref{thm1} (see Subsection \ref{proofthm1}) we relate the expected influence spread coming from the non-adaptive greedy algorithm with that of the optimal adaptive policy. We first need some notation and preliminary results. Let $G=(V=[n],E,(p_{uv})_{(u,v)\in E})$ be an influence graph, and let $k\in [n]_2$. For any $t\in [k]_0$, let $S_t$ denote the set of the first $t$ seeds selected by the greedy algorithm, so that $\sigma(S_k)$ is the expected influence spread of the solution returned by the algorithm. Let $\pi$ be an optimal adaptive policy, and let $x=(x_1,\ldots, x_n)$ be the vector such that $x_i$ is the probability that node $i$ is selected by $\pi$. As in Section~\ref{sec_example}, we resort again to a randomized non-adaptive policy to relate the expected influence spread of the greedy algorithm with that of the adaptive policy.

\subparagraph*{Randomized Non-adaptive Policy.} For any $t\in [k-1]_0$, as intermediate step of our analysis, we consider again the randomized non-adaptive policy ${\sf Rand}_t$ defined in Section~\ref{sec_example}: starting from the greedy solution $S_{t}$, we compute a random set $S_{\bm\rho,t}:=S_{t}\cup\{\bm \rho\}$, where $\bm\rho\in [n]$ is a random item such that $\mathbb{P}[\bm\rho=i]=x_i/k$ for any $i\in [n]$ and selected independently from any other event. We observe that the expected value of $f$ under ${\sf Rand}_{t}$ is $\E_{\bm \rho}[\sigma(S_{\bm\rho,t})]$. 

As a support of our analysis, we also define a new diffusion model and a hybrid adaptive policy. In particular, the new diffusion model will be used to recover certain properties connected with the submodularity that, as in Section~\ref{sec_example}, will allow us to relate the expected influence spread of the randomized non-adaptive policy with that of the hybrid adaptive policy. Furthermore, following again the approach of Section \ref{sec_example}, the hybrid adaptive policy is obtained by combining the greedy (non-adaptive) solution and the optimal adaptive one, and it will be used in our analysis to get an upper bound on the optimal adaptive spread.

\subparagraph*{2-level Diffusion Model.} In the 2-level diffusion model each selected seed $u\in V$ has two chances to influence its neighbours, and all the non-seeds have one chance only, i.e., the activation probability for all the edges $(u,v)$ is $1-(1-p_{u,v})^2$ if $u$ is a seed, and $p_{u,v}$ otherwise. More formally, let $ \hat{\bm L}$ be a live-edge graph distributed as $\bm L$ and independent from $\bm L$. Given a set of seeds $S$, the {\em 2-level live-edge graph} can be defined as $\bm L^2(S):=\bm L\cup \{\hat{\bm L}\cap \{(u,v)\in E:u\in S\}\}$. Given a set of nodes $S$, let $\sigma^2_{\bm L,\hat{\bm L}}(S):=\sigma_{\bm L^2(S)}(S)$ denote the influence spread induced by $S$ in live-edge graph $\bm L^2(S)$, and let $\sigma^2(S):=\mathbb{E}_{\bm L,\hat{\bm L}}[\sigma^2_{\bm L,\hat{\bm L}}(S)]$ denote the expected influence spread induced by $S$ under the $2$-level diffusion model. Let $ \Phi$ and $\hat{\Phi}$ denote the random realisations associated to live-edge graphs $\bm L$ and $\hat{\bm L}$, respectively.

\subparagraph*{2-Level Hybrid Adaptive Policy.}  For any $t\in [k-1]_0$, let ${\sf Hyb}^2_{t}$ be a hybrid adaptive policy defined as follows: (i) ${\sf Hyb}^2_{t}$ selects all the nodes in $S_t$ as seeds; (ii) then, ${\sf Hyb}^2_{t}$ adds to $S_t$ all the nodes that the optimal adaptive policy $\pi$ would have select when starting from the empty realisation, and observing, at each step, partial realisations coming from the live-edge graph $\hat{\bm L}$ only; in other words, for any seed $v\in V$ selected by the policy, the new set of nodes that the policy observes (to choose the next seed) is $\hat{\Phi}(v)$; (iii) finally, denoting with $\hat{\Psi}_{\pi}$ the random realisation returned by $\pi$, the expected influence spread of ${\sf Hyb}^2_{t}$ is defined as $\E_{\bm L,\hat{\bm L}}[\sigma^2_{\bm L,\hat{\bm L}}(dom(\hat{\Psi}_{\pi})\cup S_t)]$, i.e., it is the expected influence spread determined $dom(\hat{\Psi}_{\pi})\cup S_t$, according to the 2-level live-edge graph $\bm L^2(dom(\hat{\Psi}_{\pi})\cup S_t)$.

In what follows we give some technical results which are based on the above definitions and will be used in Subsection \ref{proofthm1} to show the main theorem. In particular, we show that the 2-level diffusion model satisfies certain properties connected with submodular set functions (see Subsection \ref{sub1}); then, we use such properties to relate the expected influence spread in the ordinary diffusion model to that of the 2-level diffusion model (see Subsection \ref{sub2}); finally, by using a similar approach as in Section \ref{sec_example}, we use the above relations to relate the efficiency of the randomized non-adaptive policy with that of the optimal adaptive policy, via the 2-level hybrid adaptive policy (see Subsection \ref{sub3}). 

\subsection{Adaptive Submodularity in the 2-Level Diffusion Model}\label{sub1} The adaptive submodularity \cite{Golovin2011a} is a property that extends  the well-known concept of submodularity to the adaptive framework, and allows us to design efficient adaptive approximation algorithms. In particular, the adaptive submodularity states that, given two subrealisations $\psi\subseteq \psi'$ and a node $v\in V$, adding node $v$ under partial realisation $\psi'$ causes an expected increment of the influence spread that is not higher than that caused under partial realisation $\psi$. Unfortunately, as shown in  the arXiv version of \cite{Golovin2011a}, the myopic feedback model, in general, does not satisfy the adaptive submodularity. 

Anyway, by resorting to the 2-level diffusion model, we can recover a similar property as the adaptive submodularity. Given $S\subseteq V$, a partial realisation $\hat{\psi}$, and $v\in V$, let 
\begin{multline*}
{\Delta}^2_{S}(v|\hat{\psi}):=\nonumber\\
\E_{\bm L,\hat{\bm L}}\left[\sigma_{\bm L^2(\{v\}\cup dom(\hat{\psi})\setminus S)}(\{v\}\cup S\cup dom(\hat{\psi}))-\sigma_{\bm L^2(dom(\hat{\psi})\setminus S)}(S\cup dom(\hat{\psi}))|\hat{\psi}\subseteq \hat{\Phi}\right]
\end{multline*}
denote the expected increment of the influence spread w.r.t. the $2$-level diffusion model when adding seed $v$ to the set of nodes $S\cup dom(\hat{\psi})$, but assuming that the nodes in $S$ have a unique chance to influence their neighbors, and that partial realisation $\hat{\psi}$ has been observed. We say that the 2-level diffusion model is {\em adaptive submodular} if, for any $S\subseteq V$, any partial realisations $\hat{\psi},\hat{\psi}'$ with $\hat{\psi}\subseteq \hat{\psi}'$, and any $v\in V$, we have that $\Delta^2_S(v|\hat{\psi})\geq {\Delta}^2_S(v|\hat{\psi}')$.
\begin{lemma}\label{cla3}
The 2-level diffusion model is adaptive submodular. 
\end{lemma}
\subsection{From the Ordinary to the 2-level Diffusion Model}\label{sub2}
In this subsection, we see how to relate the 2-level diffusion model to the ordinary one.
\begin{lemma}\label{cla1}
We have that $\sigma^2(S)\leq 2\cdot \sigma(S)$ for any $S\subseteq V$.

\end{lemma}
Now, given a set $S\subseteq V$ and $v\in V$, let $\Delta_S(v):=\sigma(S\cup\{v\})-\sigma(S)$, that is the expected increment under the ordinary diffusion model when adding a new node to a set $S$, and let $\Delta^2_S(v):=\E_{\bm L,\hat{\bm L}}[\sigma_{\bm L^2(\{v\})}(S\cup\{v\})-\sigma_{\bm L}(S)]$, that is the above expected increment, with the further assumption that $v$ has two chances to influence its neighbors.
\begin{lemma}\label{cla.1}
We have that $\Delta^2_S(v)\leq 2\cdot \Delta_S(v)$ for any $S\subseteq V$ and $v\in V$. 
\end{lemma}
\subsection{From the Adaptive to the Randomized Non-adaptive Policy}\label{sub3}
The following three lemmas (Lemma \ref{lem2}, Lemma \ref{lem1}, and \ref{cla0}) relate the randomized non-adaptive policy with the 2-level hybrid adaptive policy, and then with the optimal adaptive policy of the ordinary diffusion model. In particular, the proofs of Lemma \ref{lem2}, Lemma \ref{lem1}, and \ref{cla0}, resort to a similar approach of the proofs of Lemma \ref{lemprel1}, Lemma \ref{lemprel2}, and Lemma \ref{lemprel3} of Section \ref{sec_example}.

\begin{lemma}\label{lem2}
For any $S\subseteq V$, we have that $k \cdot (\E_{\bm \rho}[\sigma(\{\bm \rho\}\cup S)]-\sigma(S))\geq \sum_{v\in V\setminus S}x_v\cdot \Delta_S(v).$
\end{lemma}
The following lemma relates the adaptive setting with the non-adaptive one, and its proof resorts to the adaptive submodularity defined in Subsection \ref{sub1}. 
\begin{lemma}\label{lem1}
We have $\E_{\bm L,\hat{\bm L}}[\sigma^2_{\bm L,\hat{\bm L}}(dom(\hat{\Psi}_{\pi})\cup S)]\leq \sigma^2(S)+\sum_{v\in V\setminus S}x_v\cdot \Delta^2_S(v)$ for any $S\subseteq V$. 
\end{lemma}
The following lemma shows that the optimal adaptive influence spread is upper bounded by that expected influence spread of the 2-level hybrid adaptive policy. 
\begin{lemma}\label{cla0}
We have $OPT_A(G,k)\leq \E_{\bm L,\hat{\bm L}}[\sigma^2_{\bm L,\hat{\bm L}}(dom(\hat{\Psi}_{\pi})\cup S)]$ for any $S\subseteq V$. 
\end{lemma}
\subsection{Proof of Theorem \ref{thm1}}\label{proofthm1}
Armed with the above results, we can now prove Theorem~\ref{thm1}. For any $t\in [k]_0$, let $GR_N(G,t):=\sigma(S_t)$ denote the expected influence spread $\sigma(S_t)$ obtained when the first $t$ seeds have been selected by the greedy algorithm. We have that 
\begin{align}
&GR_N(G,t+1)-GR_N(G,t)\nonumber\\
&=\max_{v\in V}\left[\sigma(\{v\}\cup S_{t})\right]-\sigma(S_{t})\nonumber\\
&\geq \overbrace{\E_{\bm \rho}[\sigma(\{\bm \rho\}\cup S_{t})]}^{\text{Exp. value of }{\sf Rand}_t}-\sigma(S_{t})\label{equ1}\\
&\geq \frac{1}{k}\sum_{v\in V\setminus S_{t}}x_v\cdot \Delta_{S_t}(v)\label{equ2}\\
&\geq \frac{1}{2k}\sum_{v\in V\setminus S_t}x_v\cdot  \Delta^2_{S_t}(v)\label{equ2.0}\\
&\geq \frac{1}{2k} (\overbrace{\E_{\bm L,\hat{\bm L}}[\sigma^2_{\bm L,\hat{\bm L}}(dom(\hat{\Psi}_{\pi})\cup S_t)]}^{\text{Exp. value of ${\sf Hyb}^2_t$}}-\sigma^2(S_t))\label{equ3}\\
&\geq \frac{1}{2k} (OPT_A(G,k)-\sigma^2(S_t))\label{equ4}\\
&\geq \frac{1}{2k}\left(OPT_A(G,k)-2\cdot \sigma(S_t)\right)\label{equ5}\\
&=\frac{1}{2k}OPT_A(G,k)-\frac{1}{k}GR_N(G,t),\label{equfund}
\end{align}
where \eqref{equ1} holds since the greedy strategy adds to $S_t$ the node $v$ maximizing $\sigma(S_t\cup\{v\})$, \eqref{equ2} comes from Lemma \ref{lem2}, \eqref{equ2.0} comes from Lemma \ref{cla.1}, \eqref{equ3} comes from Lemma \ref{lem1}, \eqref{equ4} comes from Lemma \ref{cla0}, and  \eqref{equ5} comes from Lemma \ref{cla1}.
Thus, by \eqref{equfund} and some manipulations we get the following recursive relation: $
{GR}_N(G,t+1)\geq \frac{1}{2k}\cdot OPT_A(G,k)+\left(1-\frac{1}{k}\right)\cdot {GR}_N(G,t)$ for any $t\in [k-1]_0.$ By applying iteratively the above inequality, we get
$
{GR}_N(G,k)\geq \frac{1}{2k}\cdot \sum_{t=0}^{k-1}\left(1-\frac{1}{k}\right)^{t}\cdot OPT_A(G,k)=\frac{1}{2}\left(1-\left(1-\frac{1}{k}\right)^{k}\right)\cdot OPT_A(G,k),
$
that leads to $
\frac{{GR}_N(G,k)}{OPT_A(G,k)}\geq \frac{1}{2}\left(1-\left(1-\frac{1}{k}\right)^{k}\right) \geq \frac{1}{2}\left(1-\frac{1}{e}\right),$
and this shows the claim. \qed

\begin{remark}\label{adgapcor}
By Theorem \ref{thm1}, we can easily show that, for any influence graph $G$ with $n$ nodes, the $k$-adaptivity gap of $G$ is at most $2\left(1-\left(1-\frac{1}{k}\right)^{k}\right)^{-1}\leq \frac{2e}{e-1}\approx 3.164$. 
\end{remark}

\section{The Efficiency of the Adaptive Greedy Algorithm}\label{ad-sec}
We show that the adaptive version of the greedy algorithm guarantees an even better approximation ratio of $1-\frac{1}{\sqrt{e}}\approx 0.393$ for the adaptive influence maximization problem. The {\em adaptive greedy algorithm} is an adaptive policy $\pi^{GR}_{k}$ that selects $k$ seeds in $k$ steps, and at each step $t$ selects the $t$-th seed that maximizes the expected influence spread conditioned by the observed realisation. 
\begin{theorem}\label{thm2}
Given an influence graph $G$ with $n$ nodes and $k\in [n]_2$, the adaptive greedy algorithm is a $1-\left(1-\frac{1}{2k}\right)^{k}\geq 1-\frac{1}{\sqrt{e}}$ approximation algorithm for the adaptive influence maximization problem (under IC and  myopic feedback) applied to $(G,k)$. 
\end{theorem}
In the proof of Theorem \ref{thm2} we relate the expected influence spread coming from the adaptive greedy algorithm with that of the optimal adaptive policy, passing trough a new hybrid adaptive policy. In the same spirit of Theorem~\ref{thm1}, we also consider a new diffusion model and a new notion of adaptive submodularity (slightly different from that of Subsection \ref{sub1}). To show the main theorem (see Subsection \ref{proofthm2}) we need some notation and preliminary results. Let $G=(V=[n],E,(p_{uv})_{(u,v)\in E})$ be an influence graph, and let $k\in[n]_2$. Let $\pi$ be an optimal adaptive policy, and let $x=(x_1,\ldots, x_n)$ be the vector such that $x_i$ is the probability that node $i$ is selected by $\pi$. 

\subparagraph*{Strong 2-level Diffusion Model.} We define $\bm L$, $\hat{\bm L}$, ${\bm L}^2(S)$, $\Phi$, $\hat{ \Phi}$ as in the 2-level diffusion model considered in Section \ref{sec_inarb}. Given a partial realisation $\psi$ and a set $S\subseteq V$, let $\sigma^{2}_{\bm L,\hat{\bm L},\psi}(S):=\sigma_{\bm L^2(S\setminus dom(\psi))}(S)$ denote the influence spread induced by $S$ in live-edge graph $\bm L^2(S\setminus dom(\psi))$, and let $\sigma^{2}_{\psi}(S):=\mathbb{E}_{\bm L,\hat{\bm L}}[\sigma^{2}_{\bm L,\hat{\bm L},\psi}(S)|\psi\subseteq \Phi]$ denote the above influence spread in expectation, conditioned by partial realisation $\psi$. 

\subparagraph*{Strong 2-level Hybrid Adaptive Policy.} Given a partial realisation $\psi$, let ${\sf Hyb}_\psi^2$ be a hybrid adaptive policy defined as follows: (i) ${\sf Hyb}_\psi^2$ selects all the nodes in $dom(\psi)$ as seeds; (ii) then, ${\sf Hyb}_\psi^2$ adds to $dom(\psi)$ all the nodes that policy $\pi$ would have select when starting from the empty realisation, and observing, at each step, partial realisations coming from the live-edge graph $\hat{\bm L}$ only; (iii) finally, denoting with $\hat{\Psi}_{\pi}$ the random realisation returned by policy $\pi$, the expected influence spread of ${\sf Hyb}_\psi^2$ is defined as $\E_{\bm L,\hat{\bm L}}[\sigma^{2}_{\bm L,\hat{\bm L},\psi}(dom(\psi)\cup dom(\hat{\Psi}_{\pi}))|~\psi\subseteq~\Phi]$, i.e., it is the expected influence spread determined $dom(\psi)\cup dom(\hat{\Psi}_{\pi})$ in the strong 2-level diffusion model, conditioned by partial realisation $\psi$. Differently from the 2-level hybrid policy defined in Section \ref{sec_inarb}, the expected influence spread of the hybrid policy defined here is conditioned by partial realisation $\psi$ and the unique seeds that have two chances to influence their neighbors are those in $dom(\hat{\Psi}_{\pi})\setminus dom(\psi)$.

\subparagraph*{Strong Adaptive Submodularity of the Strong 2-level Diffusion Model.}Given two partial realisations $\psi,\hat{\psi}$  and $v\in V$, let
\begin{multline*}
{\Delta}_{\psi}^2(v|\hat{\psi})=\nonumber\\
\E_{\bm L,\hat{\bm L}}\left[\sigma^{2}_{\bm L,\hat{\bm L},\psi}(\{v\}\cup dom(\psi)\cup dom(\hat{\psi}))-\sigma^{2}_{\bm L,\hat{\bm L},\psi}(dom(\psi)\cup dom(\hat{\psi}))|\psi\subseteq \Phi,\hat{\psi}\subseteq \hat{\Phi}\right],
\end{multline*}
i.e., $\Delta^2_\psi(v|\hat{\psi})$ is the expected increment of the influence spread in the strong 2-level diffusion model when adding seed $v$ to the set of nodes $dom(\hat{\psi})\cup dom(\psi)$, conditioned by the observation of partial realisations $\hat{\psi}$ and $\psi$. We say that the strong 2-level diffusion model is {\em strongly adaptive submodular} if, for any partial realisations $\hat{\psi},\hat{\psi}'$ with $\hat{\psi}\subseteq \hat{\psi}'$, any partial realisation $\psi$, and any $v\in V$, we have that $\Delta^2_\psi(v|\hat{\psi})\geq {\Delta}^2_\psi(v|\hat{\psi}')$. 

\begin{lemma}\label{scla3}
The strong 2-level diffusion model is strongly adaptive submodular. 
\end{lemma}
\subparagraph*{From the Ordinary to the Strong 2-level Diffusion Model.}
Given a partial realisation $\psi$, and $v\in V$, let $\Delta_\psi(v):=\E_{\bm L}[\sigma_{\bm L}(\{v\}\cup dom(\psi))-\sigma_{\bm L}(dom(\psi))|\psi \subseteq \Phi]$, that is the expected increment under the ordinary diffusion model when adding a new node to the nodes in $dom(\psi)$ conditioned by partial realisation $\psi$, and let $\Delta^2_\psi(v):=\Delta^{2}_{\psi}(v|\emptyset)=\E_{\bm L,\hat{\bm L}}[\sigma_{\bm L,\hat{\bm L},\psi}(\{v\}\cup dom(\psi))-\sigma_{\bm L,\hat{\bm L},\psi}(dom(\psi))|\psi \subseteq  \Phi]$, that is the above conditional expectation, but w.r.t. the strong 2-level diffusion model. The following lemma can be shown analogously to Lemma \ref{cla.1}.
\begin{lemma}\label{scla.1}
We have that $\Delta^2_\psi(v)\leq 2\cdot \Delta_\psi(v)$ for any partial realisation $\psi$ and $v\in V$. 
\end{lemma}
\subparagraph*{From the Optimal to the Greedy Adaptive Policy.}
The following lemma will be used to relate the strong 2-level hybrid adaptive policy with the adaptive greedy policy; its proof is similar to that of Lemma \ref{lem1}, but uses the concept of strong adaptive submodularity.
\begin{lemma}\label{slem1}
For any partial realisation $\psi$, we have,
\begin{multline}
\E_{\bm L,\hat{\bm L}}[\sigma^{2}_{\bm L,\hat{\bm L},\psi}(dom(\psi)\cup dom(\hat{\Psi}_{\pi}))|~\psi\subseteq\Phi]\leq \E_{\bm L}[\sigma_{\bm L}(dom(\psi))|\psi~\subseteq~\Phi]\\
+\sum_{v\in V\setminus dom(\psi)}x_v\cdot \Delta^2_\psi(v).
\end{multline}
\end{lemma}
The following lemma shows that the optimal adaptive influence spread is upper bounded by that expected influence spread of the strong 2-level hybrid adaptive policy, and its proof is completely analogue to that of Lemma \ref{cla0}. 
\begin{lemma}\label{scla0}
We have $OPT_A(G,k)\leq \E_{\bm L,\hat{\bm L}}[\sigma^{2}_{\bm L,\hat{\bm L},\psi}(dom(\psi)\cup dom(\hat{\Psi}_{\pi}))|~\psi\subseteq\Phi]$ for any partial realisation $\psi$.
\end{lemma}
\subsection{Proof of Theorem \ref{thm2}}\label{proofthm2}
Armed with the above lemmas, and by using a similar approach as in the proof of Theorem~\ref{thm1}, we can prove Theorem~\ref{thm2}. Given $t\in [k]_0$, let $\bm S_{t}$ denote the (random) set of the first $t$ seeds selected by the adaptive greedy policy $\pi^{GR}_k$, and let $\Psi_t$ be the (random) partial realisation such that $dom(\Psi_t)=\bm S_t$ (i.e., the partial realisation observed by policy $\pi^{GR}_k$ at the end of step $t$); let $GR_A(G,t):=\E_{\bm L}[\sigma_{\bm L}(\bm S_t)]$ denote the expected influence spread of the adaptive greedy policy after selecting the first $t$ seeds. For any $t\in [k-1]_0$, we have that
\begin{align}
&GR_A(G,t+1)-GR_A(G,t)\nonumber\\
&=\E_{\bm L}[\sigma(\bm S_{t+1})-\sigma(\bm S_{t})]\nonumber\\
&= \E_{\Psi_t}\left[\max_{v\in V\setminus dom(\Psi_t)}\E_{\bm L}[\sigma_{\bm L}(\{v\}\cup dom(\Psi_t))-\sigma_{\bm L}(dom(\Psi_t))|\Psi_t]\right]\label{sequ1.0}\\
&\geq \E_{\Psi_t}\left[\sum_{v\in V\setminus dom(\Psi_t)}\frac{x_v}{k}\cdot \E_{\bm L}[\sigma_{\bm L}(\{v\}\cup dom(\Psi_t))-\sigma_{\bm L}(dom(\Psi_t))|\Psi_t]\right]\nonumber\\
&= \frac{1}{k}\cdot \E_{\Psi_t}\left[\sum_{v\in V\setminus dom(\Psi_t)}x_v\cdot \E_{\bm L}[\sigma_{\bm L}(\{v\}\cup dom(\Psi_t))-\sigma_{\bm L}(dom(\Psi_t))|\Psi_t]\right]\nonumber\\
&=  \frac{1}{k}\cdot \E_{\Psi_t}\left[\sum_{v\in V\setminus dom(\Psi_t)}x_v\cdot \Delta_{\Psi_t}(v)\right]\nonumber\\
&\geq  \frac{1}{2k}\cdot \E_{\Psi_t}\left[\sum_{v\in V\setminus dom(\Psi_t)}x_v\cdot \Delta^2_{\Psi_t}(v)\right]\label{sequ2.0}\\
&\geq \frac{1}{2k}\cdot \E_{\Psi_t}\left[\E_{\bm L,\hat{\bm L}}[\sigma^{2}_{\bm L,\hat{\bm L},\Psi_t}(dom(\Psi_t)\cup dom(\hat{\Psi}_{\pi}))|\Psi_t]-\E_{\bm L}[\sigma_{\bm L}(dom(\Psi_t))|\Psi_t]\right]\label{sequ3}\\
&\geq  \frac{1}{2k}\cdot \E_{\Psi_t}\left[OPT_A(G,k)-\E_{\bm L}[\sigma_{\bm L}(dom(\Psi_t))|\Psi_t]\right]\label{sequ4}\\
&\geq  \frac{1}{2k}\cdot OPT_A(G,k)-\frac{1}{2k}\cdot \E_{\Psi_t}\left[\E_{\bm L}[\sigma_{\bm L}(dom(\Psi_t))|\Psi_t]\right]\nonumber\\
&=\frac{1}{2k}\cdot OPT_A(G,k)-\frac{1}{2k}\cdot \E_{\bm L}(\sigma_{\bm L}(\bm S_t))\nonumber\\
&=\frac{1}{2k}\cdot OPT_A(G,k)-\frac{1}{2k}\cdot GR_A(G,t),\label{sequfund}
\end{align}
where \eqref{sequ1.0} holds since the adaptive greedy strategy at each step adds the node maximizing the expected influence spread (conditioned by the partial realisation coming from the previous selected seeds), \eqref{sequ2.0} comes from Lemma \ref{scla.1}, \eqref{sequ3} comes from Lemma \ref{slem1}, \eqref{sequ4} comes from Lemma \ref{scla0}. 
Thus, by \eqref{sequfund} and some manipulations we get the following recursive relation:
\begin{equation}\label{sfundeqthm}
{GR}_A(G,t+1)\geq \frac{1}{2k}\cdot OPT_A(G,k)+\left(1-\frac{1}{2k}\right)\cdot {GR}_A(G,t),\quad \forall t\in [k-1]_0.
\end{equation}

By applying iteratively \eqref{sfundeqthm}, we get
\begin{equation*}
{GR}_A(G,k)\geq \frac{1}{2k}\cdot \sum_{t=0}^{k-1}\left(1-\frac{1}{2k}\right)^{t}\cdot OPT_A(G,k)=1-\left(1-\frac{1}{2k}\right)^{k}\cdot OPT_A(G,k),
\end{equation*}
that leads to 
\begin{equation*}
\frac{{GR}_A(G,k)}{OPT_A(G,k)}\geq 1-\left(1-\frac{1}{2k}\right)^{k} \geq 1-\frac{1}{\sqrt{e}},
\end{equation*}
and this shows the claim.
\section{Conclusions and Future Work}\label{sec_future}
In the context of adaptive optimization, we have introduced a new approach to relate the solution provided by a simple non-adaptive greedy policy with the adaptive optimum. 
The new approach allowed us to establish better bounds for the adaptive influence maximization problem under myopic feedback, specifically we improve both the approximation ratio of the non-adaptive greedy policy and the adaptivity gap. 

Our results open several research directions in the context of influence maximization and in more general adaptive optimization settings. The approximation factor of the non-adaptive (resp. adaptive) greedy algorithm is 
between our lower bound of $\frac{1}{2}\left(1-\frac{1}{e}\right) \approx 0.316$ (resp. $1-\frac{1}{\sqrt{e}}\approx 0.393$) and
the upper bound of $\frac{e^2+1}{(e+1)^2}\approx 0.606$~\cite{Peng2019}, and the adaptivity gap is between the lower bound of $\frac{e}{e-1}\approx 1.582$~\cite{Peng2019} and our upper bound of $\frac{2e}{e-1}\approx 3.164$. The first problem left open by our result is to close these gaps. Furthermore, the techniques introduced in this paper to relate non-adaptive policies with adaptive ones might be useful to find better bounds in several variants of the adaptive influence maximization problem, like a combination of the following settings: different feedback models (e.g., the full-adoption feedback), different diffusion models (e.g., the general triggering model \cite{Kempe2003}), and different graph classes. 

Finally, as shown in Section~\ref{sec_example}, we believe that our new approach could be efficiently used to analyze non-adaptive greedy algorithms in other adaptive optimization problems, like e.g. the stochastic probing problem~\cite{Bradac19,Gupta2016,Gupta2017}.

\bibliography{dangelo}

\begin{thebibliography}{10}

\bibitem{DBLP:conf/stacs/AdamczykSW14}
Marek Adamczyk, Maxim Sviridenko, and Justin Ward.
\newblock Submodular stochastic probing on matroids.
\newblock In {\em 31st International Symposium on Theoretical Aspects of
  Computer Science, {(STACS})}, pages 29--40, 2014.

\bibitem{Asadpour16}
Arash Asadpour and Hamid Nazerzadeh.
\newblock Maximizing stochastic monotone submodular functions.
\newblock {\em Management Science}, 62(8):2374--2391, 2016.

\bibitem{DBLP:conf/wine/AsadpourNS08}
Arash Asadpour, Hamid Nazerzadeh, and Amin Saberi.
\newblock Stochastic submodular maximization.
\newblock In {\em Internet and Network Economics, 4th International Workshop,
  (WINE)}, pages 477--489, 2008.

\bibitem{Badanidiyuru2016}
Ashwinkumar Badanidiyuru, Christos Papadimitriou, Aviad Rubinstein, Lior
  Seeman, and Yaron Singer.
\newblock {Locally adaptive optimization: Adaptive seeding for monotone
  submodular functions}.
\newblock {\em Proceedings of the Annual ACM-SIAM Symposium on Discrete
  Algorithms, (SODA)}, 1:414--429, 2016.

\bibitem{Bradac19}
Domagoj Bradac, Sahil Singla, and Goran Zuzic.
\newblock (near) optimal adaptivity gaps for stochastic multi-value probing.
\newblock In {\em Approximation, Randomization, and Combinatorial Optimization.
  Algorithms and Techniques (APPROX/RANDOM)}, pages 49:1--49:21, 2019.

\bibitem{Calinescu11}
Gruia C{\u{a}}linescu, Chandra Chekuri, Martin P{\'{a}}l, and Jan
  Vondr{\'{a}}k.
\newblock Maximizing a monotone submodular function subject to a matroid
  constraint.
\newblock {\em {SIAM} J. Comput.}, 40(6):1740--1766, 2011.

\bibitem{DBLP:journals/mor/ChanF09}
Carri~W. Chan and Vivek~F. Farias.
\newblock Stochastic depletion problems: Effective myopic policies for a class
  of dynamic optimization problems.
\newblock {\em Math. Oper. Res.}, 34(2):333--350, 2009.

\bibitem{CVZ14}
Chandra Chekuri, Jan Vondr{\'{a}}k, and Rico Zenklusen.
\newblock Submodular function maximization via the multilinear relaxation and
  contention resolution schemes.
\newblock {\em {SIAM} J. Comput.}, 43(6):1831--1879, 2014.

\bibitem{DBLP:series/synthesis/2013Chen}
Wei Chen, Laks V.~S. Lakshmanan, and Carlos Castillo.
\newblock {\em Information and Influence Propagation in Social Networks}.
\newblock Synthesis Lectures on Data Management. Morgan \& Claypool, 2013.

\bibitem{Chen2019}
Wei Chen and Binghui Peng.
\newblock On adaptivity gaps of influence maximization under the independent
  cascade model with full-adoption feedback.
\newblock {\em 30th International Symposium on Algorithms and Computation,
  (ISAAC)}, pages 24:1--24:19, 2019.

\bibitem{Chen2019a}
Wei Chen, Binghui Peng, Grant Schoenebeck, and Biaoshuai Tao.
\newblock {Adaptive Greedy versus Non-adaptive Greedy for Influence
  Maximization}.
\newblock In {\em The Thirty-Fourth {AAAI} Conference on Artificial
  Intelligence, (AAAI)}, 2020.

\bibitem{Chen10}
Wei Chen, Chi Wang, and Yajun Wang.
\newblock Scalable influence maximization for prevalent viral marketing in
  large-scale social networks.
\newblock In {\em Proceedings of the 16th {ACM} {SIGKDD} International
  Conference on Knowledge Discovery and Data Mining}, pages 1029--1038, 2010.

\bibitem{Cohen14}
Edith Cohen, Daniel Delling, Thomas Pajor, and Renato~F. Werneck.
\newblock Sketch-based influence maximization and computation: Scaling up with
  guarantees.
\newblock In {\em Proceedings of the 23rd {ACM} International Conference on
  Conference on Information and Knowledge Management, (CIKM)}, pages 629--638,
  2014.

\bibitem{DPV21}
Gianlorenzo D'Angelo, Debashmita Poddar, and Cosimo Vinci.
\newblock Better bounds on the adaptivity gap of influence maximization under
  full-adoption feedback.
\newblock In {\em AAAI Conference on Artificial Intelligence, (AAAI)}, 2021.
\newblock To appear.

\bibitem{DBLP:conf/soda/DeanGV05}
Brian~C. Dean, Michel~X. Goemans, and Jan Vondr{\'{a}}k.
\newblock Adaptivity and approximation for stochastic packing problems.
\newblock In {\em Proceedings of the Sixteenth Annual {ACM-SIAM} Symposium on
  Discrete Algorithms, (SODA)}, pages 395--404, 2005.

\bibitem{DBLP:journals/mor/DeanGV08}
Brian~C. Dean, Michel~X. Goemans, and Jan Vondr{\'{a}}k.
\newblock Approximating the stochastic knapsack problem: The benefit of
  adaptivity.
\newblock {\em Math. Oper. Res.}, 33(4):945--964, 2008.

\bibitem{Domingos2001}
Pedro Domingos and Matt Richardson.
\newblock {Mining the network value of customers}.
\newblock {\em Proceedings of the Seventh ACM SIGKDD International Conference
  on Knowledge Discovery and Data Mining}, pages 57--66, 2001.

\bibitem{DBLP:conf/latin/GoemansV06}
Michel~X. Goemans and Jan Vondr{\'{a}}k.
\newblock Stochastic covering and adaptivity.
\newblock In {\em Theoretical Informatics, 7th Latin American Symposium,
  (LATIN)}, pages 532--543, 2006.

\bibitem{Goldberg2013}
Sharon Goldberg and Zhenming Liu.
\newblock The diffusion of networking technologies.
\newblock In {\em Proceedings of the Twenty-Fourth Annual {ACM-SIAM} Symposium
  on Discrete Algorithms, (SODA)}, pages 1577--1594, 2013.

\bibitem{Golovin2011a}
Daniel Golovin and Andreas Krause.
\newblock {Adaptive submodularity: Theory and applications in active learning
  and stochastic optimization}.
\newblock {\em Journal of Artificial Intelligence Research}, 42:427--486, 2011.

\bibitem{Goyal11}
Amit Goyal, Wei Lu, and Laks V.~S. Lakshmanan.
\newblock {CELF++:} optimizing the greedy algorithm for influence maximization
  in social networks.
\newblock In {\em Proceedings of the 20th International Conference on World
  Wide Web, (WWW)}, pages 47--48, 2011.

\bibitem{DBLP:conf/icml/GuilloryB10}
Andrew Guillory and Jeff~A. Bilmes.
\newblock Interactive submodular set cover.
\newblock In {\em Proceedings of the 27th International Conference on Machine
  Learning, (ICML)}, pages 415--422, 2010.

\bibitem{Gupta2016}
Anupam Gupta, Viswanath Nagarajan, and Sahil Singla.
\newblock {Algorithms and Adaptivity Gaps for Stochastic Probing}.
\newblock {\em Proceedings of the twenty-seventh annual ACM-SIAM symposium on
  Discrete algorithms, (SODA)}, 151:1731--1747, 2016.

\bibitem{Gupta2017}
Anupam Gupta, Viswanath Nagarajan, and Sahil Singla.
\newblock {Adaptivity gaps for stochastic probing: Submodular and XOS
  functions}.
\newblock {\em Proceedings of the Annual ACM-SIAM Symposium on Discrete
  Algorithms, (SODA)}, pages 1688--1702, 2017.

\bibitem{Han2018a}
Kai Han, Keke Huang, Xiaokui Xiao, Jing Tang, Aixin Sun, and Xueyan Tang.
\newblock {Efficient algorithms for adaptive influence maximization}.
\newblock {\em Proceedings of the VLDB Endowment}, 11(9):1029--1040, 2018.

\bibitem{DBLP:conf/ciac/HellersteinKL15}
Lisa Hellerstein, Devorah Kletenik, and Patrick Lin.
\newblock Discrete stochastic submodular maximization: Adaptive vs.
  non-adaptive vs. offline.
\newblock In {\em Algorithms and Complexity - 9th International Conference,
  (CIAC)}, pages 235--248, 2015.

\bibitem{Kempe2003}
David Kempe, Jon Kleinberg, and {\'{E}}va Tardos.
\newblock Maximizing the spread of influence through a social network.
\newblock In {\em Proceedings of the Ninth {ACM} {SIGKDD} International
  Conference on Knowledge Discovery and Data Mining}, pages 137--146. {ACM},
  2003.

\bibitem{Kempe2015a}
David Kempe, Jon Kleinberg, and {\'{E}}va Tardos.
\newblock {Maximizing the spread of influence through a social network}.
\newblock {\em Theory of Computing}, 11:105--147, 2015.

\bibitem{Leskovec2007}
Jure Leskovec, Andreas Krause, Carlos Guestrin, Christos Faloutsos, Jeanne~M.
  VanBriesen, and Natalie~S. Glance.
\newblock Cost-effective outbreak detection in networks.
\newblock In {\em Proceedings of the 13th {ACM} {SIGKDD} International
  Conference on Knowledge Discovery and Data Mining}, pages 420--429, 2007.

\bibitem{Li2018}
Y.~{Li}, J.~{Fan}, Y.~{Wang}, and K.~{Tan}.
\newblock Influence maximization on social graphs: A survey.
\newblock {\em IEEE Transactions on Knowledge and Data Engineering},
  30(10):1852--1872, 2018.

\bibitem{Nguyen16}
Hung~T. Nguyen, My~T. Thai, and Thang~N. Dinh.
\newblock Stop-and-stare: Optimal sampling algorithms for viral marketing in
  billion-scale networks.
\newblock In {\em Proceedings of the 2016 International Conference on
  Management of Data, (SIGMOD)}, pages 695--710, 2016.

\bibitem{DBLP:journals/corr/abs-1803-07639}
Srinivasan Parthasarathy.
\newblock An analysis of the greedy algorithm for stochastic set cover.
\newblock {\em CoRR}, abs/1803.07639, 2018.

\bibitem{Peng2019}
Binghui Peng and Wei Chen.
\newblock Adaptive influence maximization with myopic feedback.
\newblock {\em Advances in Neural Information Processing Systems 32: Annual
  Conference on Neural Information Processing System, (NeurIPS)}, pages
  5575--5584, 2019.

\bibitem{Richardson2002}
Matthew Richardson and Pedro~M. Domingos.
\newblock Mining knowledge-sharing sites for viral marketing.
\newblock In {\em Proceedings of the Eighth {ACM} {SIGKDD} International
  Conference on Knowledge Discovery and Data Mining}, pages 61--70, 2002.

\bibitem{Rubinstein2015}
Aviad Rubinstein, Lior Seeman, and Yaron Singer.
\newblock {Approximability of adaptive seeding under knapsack constraints}.
\newblock {\em Proceedings of the 2015 ACM Conference on Economics and
  Computation, (EC)}, pages 797--814, 2015.

\bibitem{Salha2018}
Guillaume Salha, Nikolaos Tziortziotis, and Michalis Vazirgiannis.
\newblock {Adaptive submodular influence maximization with myopic feedback}.
\newblock {\em Proceedings of the 2018 IEEE/ACM International Conference on
  Advances in Social Networks Analysis and Mining (ASONAM)}, pages 455--462,
  2018.

\bibitem{Seeman2013}
Lior Seeman and Yaron Singer.
\newblock {Adaptive seeding in social networks}.
\newblock {\em Proceedings - Annual IEEE Symposium on Foundations of Computer
  Science, (FOCS)}, pages 459--468, 2013.

\bibitem{Singer2016}
Yaron Singer.
\newblock {Influence maximization through adaptive seeding}.
\newblock {\em ACM SIGecom Exchanges}, 15(1):32--59, 2016.

\bibitem{Sun2018}
Lichao Sun, Weiran Huang, Philip~S. Yu, and Wei Chen.
\newblock {Multi-round influence maximization}.
\newblock {\em Proceedings of the ACM SIGKDD International Conference on
  Knowledge Discovery and Data Mining}, pages 2249--2258, 2018.

\bibitem{Tang2019}
Jing Tang, Laks~V.S. Lakshmanan, Keke Huang, Xueyan Tang, Aixin Sun, Xiaokui
  Xiao, and Andrew Lim.
\newblock {Efficient approximation algorithms for adaptive seed minimization}.
\newblock {\em Proceedings of the ACM SIGMOD International Conference on
  Management of Data}, pages 1096--1113, 2019.

\bibitem{Tang15}
Youze Tang, Yanchen Shi, and Xiaokui Xiao.
\newblock Influence maximization in near-linear time: {A} martingale approach.
\newblock In {\em Proceedings of the 2015 {ACM} {SIGMOD} International
  Conference on Management of Data}, pages 1539--1554, 2015.

\bibitem{Tang2014}
Youze Tang, Xiaokui Xiao, and Yanchen Shi.
\newblock {Influence Maximization : Near-Optimal Time Complexity Meets
  Practical Efficiency}.
\newblock {\em SIGMOD '14: Proceedings of the 2014 ACM SIGMOD International
  Conference on Management of Data}, pages 75--86, 2014.

\bibitem{Tong2019}
Guangmo Tong and Ruiqi Wang.
\newblock {Adaptive Influence Maximization under General Feedback Models}.
\newblock In {\em arXiv: 1902.00192v3}, 2019.

\bibitem{Tong2015}
Guangmo Tong, Ruiqi Wang, Zheng Dong, and Xiang Li.
\newblock {Time-constrained Adaptive Influence Maximization}.
\newblock In {\em arXiv: 2001.01742v2}, 2020.

\bibitem{Tong2017}
Guangmo Tong, Weili Wu, Shaojie Tang, and Ding~Zhu Du.
\newblock {Adaptive Influence Maximization in Dynamic Social Networks}.
\newblock {\em IEEE/ACM Transactions on Networking}, 25(1):112--125, 2017.

\bibitem{DBLP:journals/corr/VaswaniL16}
Sharan Vaswani and Laks V.~S. Lakshmanan.
\newblock Adaptive influence maximization in social networks: Why commit when
  you can adapt?
\newblock {\em CoRR}, abs/1604.08171, 2016.

\bibitem{Yuan2017}
Jing Yuan and Shaojie Tang.
\newblock {No time to observe: Adaptive influence maximization with partial
  feedback}.
\newblock {\em IJCAI International Joint Conference on Artificial
  Intelligence}, pages 3908--3914, 2017.

\end{thebibliography}
\newpage
\appendix
\section{Missing Proofs from Section \ref{sec_example}}
\subsection*{Proof of Lemma \ref{lemprel1}}
As $\mathbb{P}[\bm \rho=i]=\frac{x_i}{k}$ for any $i\in [n]$, we have that 
\begin{align*}
&\E_{\bm \theta,\bm \rho}[f(\bm\theta(S\cup\{\bm \rho\}))]-\E_{\bm \theta}[f(\bm\theta(S))]\\
&=k\cdot \E_{\bm \rho}[\E_{\bm\theta}[f(\bm\theta(S\cup{\bm \rho}))-f(\bm\theta(S))]]\\
&=k\cdot \sum_{i\in [n]\setminus S} \mathbb{P}[{\bm \rho}=i]\cdot \E_{\bm\theta}[\Delta(i|\bm\theta(S))]\\
&=k\cdot \sum_{i\in [n]\setminus S} \frac{x_i}{k}\cdot \E_{\bm\theta}[\Delta(i|\bm\theta(S))]\\
&=\sum_{i\in [n]\setminus S} x_i\cdot \E_{\bm\theta}[\Delta(i|\bm\theta(S))].
\end{align*}
\subsection*{Proof of Lemma \ref{lemprel2}}
It is sufficient showing that $\E_{\hat{\bm \theta}}[f(\xi\vee\hat{\bm \theta}(U_{\hat{\bm\theta},k}(\pi)\cup S))]\leq \E_{\hat{\bm \theta}}[f(\xi\vee\hat{\bm \theta}(S)]+\sum_{i\in [n]\setminus S}x_i\cdot \Delta(i|\xi)$ for any partial state $\xi\sim \bm\theta(S)$, so that, by doing the expectation on $\xi\sim\bm\theta(S)$, the claim follows. 

We observe that $\E_{\hat{\bm \theta}}[f(\xi\vee\hat{\bm \theta}(U_{\hat{\bm\theta},k}(\pi)\cup S))]$ can be rewritten as the expected value $\E_{\hat{\bm \theta}}[f(\xi\vee\hat{\bm \theta}(S))]$ coming from the selection of $S$, plus the sum of expected increments of $f$ caused by all the nodes $i\in [n]\setminus S$ selected by policy $\pi$, when observing partial states from $\hat{\bm \theta}$. We have that the second sum can be written as  $$\sum_{i\in [n]\setminus S}\sum_{\hat{\xi}}\chi(i|\hat{\xi})\cdot \Delta(i|\xi\vee \hat{\xi}),$$ where $\chi(i|\hat{\xi})\in \{0,1\}$ is the indicator random variable that is equal to $1$ if policy $\pi$ visits state $\hat{\xi}$ at some step of the execution and then selects node $i$ (i.e., $i=\pi(\hat{\xi})$), and $\chi(i|\hat{\xi})=0$ otherwise. Thus,
\begin{align}
&\E_{\hat{\bm \theta}}[f(\xi\vee\hat{\bm \theta}(U_{\hat{\bm\theta},k}(\pi)\cup S))]\nonumber\\
&=\E_{\hat{\bm \theta}}[f(\xi\vee\hat{\bm \theta}(S))]+\sum_{i\in [n]\setminus S}\sum_{\hat{\xi}}\chi(i|\hat{\xi})\cdot \Delta(i|\xi\vee \hat{\xi}).\label{xi1}
\end{align}
For any partial state $\hat{\xi}$ we have that
\begin{align}
&\Delta(i|\xi\vee \hat{\xi})]\nonumber\\
&=\E_{\bm e^i}[f(\xi \vee \hat{\xi}\vee \bm e^{i})-f(\xi\vee \hat{\xi})]\nonumber\\
&=\E_{\bm e^i}[f((\xi \vee \hat{\xi})\vee (\xi\vee \bm e^{i}))-f(\xi\vee \hat{\xi})]\nonumber\\
&\leq \E_{\bm e^i}[f(\xi \vee \bm e^{i})-f((\xi\vee \hat{\xi})\wedge (\xi\vee \bm e^{i}))]\label{xi1.1}\\
&\leq  \E_{\bm e^i}[f(\xi \vee \bm e^{i})-f(\xi)]\nonumber\\
&=\Delta(i|\xi),\label{xi2}
\end{align}
where \eqref{xi1.1} holds since $f$ is a monotone submodular value function: indeed, we have that $f(x\vee y)+f(x\wedge y)\leq f(x)+f(y)$ with $x:=(\xi \vee \hat{\xi})$ and $y:= (\xi\vee \bm e^{i})$, thus showing \eqref{xi1.1}. Furthermore, by exploiting the fact that $x_i$ is the probability that an item $i$ is selected in some step of policy $\pi$, we get
\begin{equation}\label{xi3}
\sum_{\hat{\xi}}\chi(i|\hat{\xi})=x_i.
\end{equation}
Finally, by combining \eqref{xi1}, \eqref{xi2}, and \eqref{xi3}, we get
\begin{align*}
&\E_{\hat{\bm \theta}}[f(\xi\vee\hat{\bm \theta}(U_{\hat{\bm\theta},k}(\pi)\cup S))]\\
&=\E_{\hat{\bm \theta}}[f(\xi\vee\hat{\bm \theta}(S))]+\sum_{i\in [n]\setminus S}\sum_{\hat{\xi}}\chi(i|\hat{\xi})\cdot \Delta(i|\xi\vee \hat{\xi})\\
&\leq \E_{\hat{\bm \theta}}[f(\xi\vee\hat{\bm \theta}(S))]+\sum_{i\in [n]\setminus S}\sum_{\hat{\xi}}\chi(i|\hat{\xi})\cdot \Delta(i|\xi)\\
&=\E_{\hat{\bm \theta}}[f(\xi\vee\hat{\bm \theta}(S))]+\sum_{i\in [n]\setminus S}\left(\sum_{\hat{\xi}}\chi(i|\hat{\xi})\right)\cdot \Delta(i|\xi)\\
&=\E_{\hat{\bm \theta}}[f(\xi\vee\hat{\bm \theta}(S))]+\sum_{i\in [n]\setminus S}x_i\cdot \Delta(i|\xi),
\end{align*}
thus showing the claim. 
\subsection*{Proof of Lemma \ref{lemprel3}}
The claim easily follows from the fact that $\hat{\bm \theta}(U_{\hat{\bm\theta},k}(\pi))$ is componentwise non-higher than $\bm\theta(S)\vee \hat{\bm \theta}(U_{\hat{\bm\theta},k}(\pi)\cup S)$, thus, as $f$ is monotone we get $OPT_A(k)=\E_{\hat{\bm \theta}}[f(\hat{\bm \theta}(U_{\hat{\bm\theta},k}(\pi)))]\leq \E_{\bm\theta, \hat{\bm \theta}}[f(\bm\theta(S)\vee \hat{\bm \theta}(U_{\hat{\bm\theta},k}(\pi)\cup S))].$ 
\section{Missing Proofs from Section \ref{sec_prel}}
\subsection*{Proof of Lemma \ref{cla3}}
To show the lemma, we connect the adaptive submodularity of the 2-level diffusion model, with the (non-adaptive) submodularity of the ordinary diffusion model shown by Kempe~et~al.~\cite{Kempe2015a}. They observed that, for any live-edge graph $L$,  the expected influence spread function $\sigma_L$ is a submodular set function: for any $U,Y,Z\subseteq V$ such that $U\subseteq Y$ we have that 
\begin{equation}\label{submod}
\sigma_L(U\cup Z)-\sigma_L(U)\geq \sigma_L(Y\cup Z)-\sigma_L(Y).
\end{equation}
Fix $S\subseteq V$, two partial realisations $\hat{\psi},\hat{\psi}'$ with $\hat{\psi}\subseteq \hat{\psi}'$, and $v\in V$. Given a partial realisation $\overline{\psi}$, let $\overline{\psi}_{\setminus S}$ denote the subrealisation of $\overline{\psi}$ with $dom(\overline{\psi}_{\setminus S})=dom(\overline{\psi})\setminus S$. We~get
\begin{align}
&\Delta^2_S(v|\hat{\psi})\nonumber\\
&=\E_{\bm L,\hat{\bm L}}\left[\sigma_{\bm L^2(\{v\}\cup dom(\hat{\psi})\setminus S)}(\{v\}\cup S\cup dom(\hat{\psi}))-\sigma_{\bm L^2(dom(\hat{\psi})\setminus S)}(S\cup dom(\hat{\psi}))|\hat{\psi}\subseteq \hat{\Phi}\right]\nonumber\\
&=\E_{\hat{\Phi}(v)}\left[\E_{\bm L}[\sigma_{\bm L}(\hat{\Phi}(v)\cup S\cup Im(\hat{\psi}_{\setminus S}))-\sigma_{\bm L}(S\cup  Im(\hat{\psi}_{\setminus S}))]\right]\nonumber\\
&\geq \E_{\hat{\Phi}(v)}\left[\E_{\bm L}[\sigma_{\bm L}(\hat{\Phi}(v)\cup S\cup Im(\hat{\psi}'_{\setminus S}))-\sigma_{\bm L}( S\cup Im(\hat{\psi}'_{\setminus S}))]\right]\label{adsub}\\
&=\Delta_S^2(v|\hat{\psi}'),\nonumber
\end{align}
where \eqref{adsub} holds by non-adaptive submodularity: for $U:=S\cup Im(\hat{\psi}_{\setminus S})$, $Y:=S\cup Im(\hat{\psi}'_{\setminus S})$, and $Z:=\hat{\Phi}(v)$, we have that $U\subseteq Y$, thus we can apply \eqref{submod} to derive \eqref{adsub}. 
\subsection*{Proof of Lemma \ref{cla1}}
By using the fact that the random graphs $\bm L$ and $\hat{\bm L}$ are independent and identically distributed, we have that 
\begin{align*}
&\sigma^2(S)=\mathbb{E}_{\bm L^2(S)}[\sigma_{\bm L^2(S)}(S)]\\
&\leq \mathbb{E}_{\bm L,\bm \hat{\bm L}}[\sigma_{\bm L}(S)+ \sigma_{\{{\bm L}\setminus \{(u,v)\in E:u\in S\}\}\cup \{{\hat{\bm L}}\cap \{(u,v)\in E:u\in S\}\}}(S)]\\
&\leq \mathbb{E}_{\bm L}[\sigma_{\bm L}(S)]+\mathbb{E}_{\bm \hat{\bm L}}[\sigma_{\hat{\bm L}}(S)]=2\cdot \sigma(S).
\end{align*}
\subsection*{Proof of Lemma \ref{cla.1}}
Let $S\subseteq V$ and $v\in V$. Given a live-edge graph $L$ and a set $S\subseteq [n]$, let $L_{\setminus S}$ be another live edge graph obtained from $L$ after removing all the nodes that would have been activated if $S$ was the set of seeds, and after removing all the edges adjacent to such activated nodes (we observe that, if $v\notin L_{\setminus S}$, then $\sigma_L(\{v\})=0$). We have that 
\begin{align}
&\Delta^2_S(v)\nonumber\\
&= \E_{\bm L^2(\{v\})_{\setminus S}}[\sigma_{\bm L^2(\{v\})_{\setminus S}}(\{v\})]\nonumber\\
&\leq 2\cdot \E_{\bm L_{\setminus S}}[\sigma_{\bm L_{\setminus S}}(\{v\})]\label{double}\\
&=2\cdot \Delta_S(v),\nonumber
\end{align}
where \eqref{double} can be shown by using analogous arguments as in the proof of Lemma \ref{cla1}. 
\subsection*{Proof of Lemma \ref{lem2}}
We have
\begin{align}
&k \cdot (\E_{\bm \rho}[\sigma(\{\bm \rho\}\cup S)]-\sigma(S))\nonumber\\
&=k \cdot (\E_{\bm \rho}[\E_{\bm L}[\sigma(\{\bm \rho\}\cup S)-\sigma(S)]])\nonumber\\
& =k\cdot \E_{\bm \rho}[\Delta_S(\bm \rho )]\nonumber\\
&=k\cdot \sum_{v\in V\setminus S}\frac{x_v}{k} \cdot \Delta_S(v)\nonumber\\
&= \sum_{v\in V\setminus S}{x_v}\cdot \Delta_S(v)\nonumber.
\end{align}
\subsection*{Proof of Lemma \ref{lem1}}
We observe that $\E_{\bm L,\hat{\bm L}}[\sigma^2_{\bm L,\hat{\bm L}}(dom(\hat{\Psi}_{\pi})\cup S)]$ can be rewritten as the expected influence spread $\sigma^2(S)$ coming from the selection of $S$ under the $2$-level diffusion model, plus the sum of the expected increments of the influence spread under the $2$-level diffusion model caused by all the nodes $v\in V\setminus S$ selected by policy $\pi$, when observing partial realisations distributed according to $\hat{\Phi}$. We observe that the second sum is at most equal to $$\sum_{v\in V\setminus S}\sum_{\hat{\psi}'}\chi(v|\hat{\psi}')\cdot \Delta^2_S(v|\hat{\psi}'),$$ where $\chi(v|\hat{\psi}')\in \{0,1\}$ is the indicator random variable that is equal to $1$ if and only if policy $\pi$ visits partial realisation $\hat{\psi}'$ at some step of the execution and then selects node $i$ (i.e., $i=\pi(\hat{\psi}')$). Thus, we have that
\begin{align}
\E_{\bm L,\hat{\bm L}}[\sigma^2_{\bm L,\hat{\bm L}}(dom(\hat{\Psi}_{\pi})\cup S)]\leq \sigma^2(S)+\sum_{v\in V\setminus S}\sum_{\hat{\psi}'}\chi(v|\hat{\psi}')\cdot \Delta^2_S(v|\hat{\psi}').\label{ciao1}
\end{align}
As the probability that each node $v\in V$ is selected by policy $\pi$ is $x_v$, we have that 
\begin{equation}\label{ciao2}
\sum_{\hat{\psi}'}\chi(v|\hat{\psi}')=x_v.
\end{equation}
We have that,
\begin{align}
&\sum_{v\in V\setminus S}\sum_{\hat{\psi}'}\chi(v|\hat{\psi}')\cdot \Delta^2_S(v|\hat{\psi}')\nonumber\\
&\leq \sum_{v\in V\setminus S}\sum_{\hat{\psi}'}\chi(v|\hat{\psi}')\cdot \Delta^2_S(v|\emptyset)\label{ciao3}\\
&\leq \sum_{v\in V\setminus S}\left(\sum_{\hat{\psi}'}\chi(v|\hat{\psi}')\right)\cdot \Delta^2_S(v|\emptyset)\nonumber\\
&=\sum_{v\in V\setminus S}x_v\cdot \Delta^2_S(v),\label{ciao4}
\end{align}
where (\ref{ciao3}) follows from Lemma \ref{cla3} (i.e., from the adaptive submodularity, and $\hat{\psi}:=\emptyset$ is the empty partial realisation), and \eqref{ciao4} holds by \eqref{ciao2}.
By applying \eqref{ciao4} to \eqref{ciao1}, we get
\begin{align*}
&\E_{\bm L,\hat{\bm L}}[\sigma^2_{\bm L,\hat{\bm L}}(dom(\hat{\Psi}_{\pi})\cup S)]\\
&\leq\sigma^2(S)+\sum_{v\in V\setminus S}\sum_{\hat{\psi}'}\chi(v|\hat{\psi}')\cdot \Delta^2_S(v|\hat{\psi}')\\
&\leq \sigma^2(S)+\sum_{v\in V\setminus S}x_v\cdot \Delta^2_S(v),
\end{align*}
and this concludes the proof. 
\subsection*{Proof of Lemma \ref{cla0}}
We have that $OPT_A(G,k)=\E_{\bm L,\hat{\bm L}}[\sigma^2_{\bm L,\hat{\bm L}}(dom(\hat{\Psi}_{\pi}))]\leq \E_{\bm L,\hat{\bm L}}[\sigma^2_{\bm L,\hat{\bm L}}(dom(\hat{\Psi}_{\pi})\cup S)].$
\section{Missing Proofs from Section \ref{ad-sec}}
\subsection*{Proof of Lemma \ref{scla3}}
This lemma can be shown analogously to Lemma \ref{cla3}. Fix two partial realisations $\hat{\psi},\hat{\psi}'$ with $\hat{\psi}\subseteq \hat{\psi}'$, a partial realisation $\psi$, and $v\in V$.  Given a partial realisation $\overline{\psi}$, let $\overline{\psi}_{\setminus \psi}$ denote the subrealisation of $\overline{\psi}$ with $dom(\overline{\psi}_{\setminus \psi})=dom(\overline{\psi})\setminus dom(\psi)$.We~get
\begin{align}
&{\Delta}_{\psi}^2(v|\hat{\psi})\nonumber\\
&=\E_{\hat{\Phi}(v)}\left[\E_{\bm L}[\sigma_{\bm L}(\hat{\Phi}(v)\cup Im(\hat{\psi}_{\setminus \psi})\cup dom(\psi))-\sigma_{\bm L}(Im(\hat{\psi}_{\setminus \psi})\cup dom(\psi))|\psi\subseteq \Phi]\right]\nonumber\\
&\geq \E_{\hat{\Phi}(v)}\left[\E_{\bm L}[\sigma_{\bm L}(\hat{\Phi}(v)\cup Im(\hat{\psi}'_{\setminus \psi})\cup dom(\psi))-\sigma_{\bm L}(Im(\hat{\psi}'_{\setminus \psi})\cup dom(\psi))|\psi\subseteq \Phi]\right]\label{snonsub}\\
&=\Delta^2_\psi(v|\hat{\psi}'),\nonumber
\end{align}
where \eqref{snonsub} holds since $\sigma_{\bm L}$ is a submodular set function. 
\subsection*{Proof of Lemma \ref{slem1}}
We observe that $\E_{\bm L,\hat{\bm L}}[\sigma^{2}_{\bm L,\hat{\bm L},\psi}(dom(\psi)\cup dom(\hat{\Psi}_{\pi}))|~\psi\subseteq\Phi]$ can be rewritten as the expected influence spread $\E_{\bm L}[\sigma_{\bm L}(dom(\psi))|\psi~\subseteq~\Phi]$ conditioned by $\psi$, plus the sum of expected increments of the influence spread under the strong $2$-level diffusion model caused by all the nodes $v\in V\setminus dom(\psi)$ selected by policy $\pi$, conditioned by $\psi$. We observe that the second sum can be written as $$\sum_{v\in V\setminus dom(\psi)}\sum_{\hat{\psi}'}\chi(v|\hat{\psi}')\cdot \Delta^2_\psi(v|\hat{\psi}'),$$ 
where $\chi(v|\hat{\psi}')\in \{0,1\}$ is the indicator random variable defined as in the proof of Lemma \ref{lem1} (i.e., equal to $1$ if and only if policy $\pi$ visits partial realisation $\hat{\psi}'$ at some step of the execution and then selects node $i$). 
Thus, we have that
\begin{align}
&\E_{\bm L,\hat{\bm L}}[\sigma^{2}_{\bm L,\hat{\bm L},\psi}(dom(\psi)\cup dom(\hat{\Psi}_{\pi}))|~\psi\subseteq\Phi]\nonumber\\
&=\E_{\bm L}[\sigma_{\bm L}(dom(\psi))|\psi~\subseteq~\Phi]+\sum_{v\in V\setminus dom(\psi)}\sum_{\hat{\psi}'}\chi(v|\hat{\psi}')\cdot \Delta^2_\psi(v|\hat{\psi}').\label{sciao1}
\end{align}
As the probability that each node $v\in V$ is selected by policy $\pi$ is $x_v$, we have that 
\begin{equation}\label{sciao2}
\sum_{\hat{\psi}'}\chi(v|\hat{\psi}')=x_v.
\end{equation}
Thus, we have that
\begin{align}
&\sum_{v\in V\setminus dom(\psi)}\sum_{\hat{\psi}'}\chi(v|\hat{\psi}')\cdot \Delta^2_\psi(v|\hat{\psi}')\nonumber\\
&\leq \sum_{v\in V\setminus dom(\psi)}\sum_{\hat{\psi}'}\chi(v|\hat{\psi}')\cdot \Delta^2_\psi(v|\emptyset)\label{sciao3}\\
&= \sum_{v\in V\setminus dom(\psi)}\left(\sum_{\hat{\psi}'}\chi(v|\hat{\psi}')\right)\cdot \Delta^2_\psi(v|\emptyset)\nonumber\\
&=\sum_{v\in V\setminus dom(\psi)}x_v\cdot \Delta^2_\psi(v),\label{sciao4}
\end{align}
where (\ref{sciao3}) follows from Lemma \ref{scla3} (i.e., from the strong adaptive submodularity), and \eqref{sciao4} holds by \eqref{sciao2}.
By applying \eqref{sciao4} to \eqref{sciao1}, we get
\begin{align*}
&\E_{\bm L,\hat{\bm L}}[\sigma^{2}_{\bm L,\hat{\bm L},\psi}(dom(\psi)\cup dom(\hat{\Psi}_{\pi}))|~\psi\subseteq\Phi]\nonumber\\
&=\E_{\bm L}[\sigma_{\bm L}(dom(\psi))|\psi~\subseteq~\Phi]+\sum_{v\in V\setminus dom(\psi)}\sum_{\hat{\psi}'}\chi(v|\hat{\psi}')\cdot \Delta^2_\psi(v|\hat{\psi}')\\
&\leq \E_{\bm L}[\sigma_{\bm L}(dom(\psi))|\psi~\subseteq~\Phi]+\sum_{v\in V\setminus dom(\psi)}x_v\cdot \Delta^2_\psi(v),
\end{align*}
thus the claim follows. 
\end{document}